\documentclass[journal,onecolumn]{IEEEtran}
\usepackage{graphicx,amsthm,amsmath,amsfonts,latexsym,amssymb}
\usepackage{color,array,verbatim,algpseudocode,bm,multirow}
\definecolor{menucolor}{rgb}{0.1,0.52,0.47}
\definecolor{urlcolor}{rgb}{0.85,0.37,0.01}
\definecolor{runcolor}{rgb}{0.46,0.44,0.701}
\definecolor{filecolor}{rgb}{0.2,0.5,0.01}
\definecolor{linkcolor}{rgb}{0.12,0.47,0.70}
\definecolor{citecolor}{rgb}{0.55,0.36,0.01}
\definecolor{anchorcolor}{rgb}{0.4,0.4,0.4}
\usepackage[colorlinks=true, urlcolor=urlcolor, linkcolor=linkcolor,citecolor=citecolor,filecolor=filecolor, anchorcolor=anchorcolor,menucolor=menucolor]{hyperref}
\usepackage{nameref}
\usepackage{zref-xr}
\usepackage{cite}
\zxrsetup{
  tozreflabel=false,
  toltxlabel=true,  
}
\usepackage[shortlabels]{enumitem}
\ifCLASSOPTIONcompsoc
  \usepackage[caption=false,font=normalsize,labelfont=sf,textfont=sf]{subfig}
\else
\usepackage[caption=false,font=footnotesize]{subfig}
\fi
\newtheorem{theorem}{Theorem}

\usepackage{lipsum}
\usepackage[mathscr]{euscript}
\usepackage{amsmath,amssymb}
\usepackage{xcolor,array,verbatim,algpseudocode,bm}
\definecolor{blueee}{rgb}{0.2, 0.2, 0.6}
\usepackage[colorlinks=true, urlcolor=urlcolor, linkcolor=linkcolor,citecolor=citecolor,filecolor=filecolor, anchorcolor=anchorcolor,menucolor=menucolor]{hyperref}
\usepackage[justification=centering]{caption}
\newcommand{\mX}{\mathcal X}

\newcommand{\mY}{\mathcal Y}
\newcommand{\mS}{\mathcal S}
\newcommand{\mQ}{\mathcal Q}
\newcommand{\mB}{\mathcal B}

\newcommand{\mZ}{\mathcal Z}
\newcommand{\mM}{\mathcal M}
\newcommand{\mL}{\mathcal L}
\newcommand{\mE}{\mathcal E}

\newcommand{\mV}{\mathcal V}

\newcommand{\mW}{\mathcal W}

\newcommand{\mbR}{\mathbb R}
\newcommand{\mbE}{\mathbb E}

\newcommand\sdots{\!\hbox to 1em{.\hss.\hss.}}

\newcommand{\mhM}{\mathcal{ \hat M}}

\newcommand{\mhB}{\mathcal{ \hat B}}

\newcommand{\bX}{\boldsymbol{X}}

\newcommand{\bx}{\boldsymbol{x}}

\newcommand{\bz}{\boldsymbol{z}}
\newcommand{\bbm}{\boldsymbol{m}}
\newcommand{\be}{\boldsymbol{e}}
\newcommand{\bbi}{\boldsymbol{i}}
\newcommand{\bj}{\boldsymbol{j}}
\newcommand{\bk}{\boldsymbol{k}}
\newcommand{\bl}{\boldsymbol{l}}
\newcommand{\bzero}{\boldsymbol{0}}
\newcommand{\rtX}{\boldsymbol{ \mathcal{   X}}}
\newcommand{\rtZ}{\boldsymbol{ \mathcal{   Z}}}
\newcommand{\rtY}{\boldsymbol{ \mathcal{   Y}}}
\newcommand{\rtU}{\boldsymbol{ \mathcal{   U}}}
\newcommand{\rtS}{\boldsymbol{ \mathcal{   S}}}
\newcommand{\rmX}{\boldsymbol{ {\sf   X}}}
\newcommand{\rmS}{\boldsymbol{ {\sf   S}}}

\newcommand{\rvX}{\boldsymbol{X}}

\newcommand{\rvZ}{\boldsymbol{Z}}
\newcommand{\sfX}{{\sf X}}
\newcommand{\sfQ}{{\sf Q}}

\newcommand{\sfA}{{\sf A}}

\newcommand{\sfI}{{\sf I}}
\newcommand{\sfS}{{\sf S}}
\newcommand{\sfB}{{\sf B}}
\newcommand{\sfC}{{\sf C}}
\newcommand{\sfD}{{\sf D}}
\DeclareMathOperator{\tr}{tr}

\DeclareMathOperator{\vecc}{vec}

\DeclareMathOperator{\N}{\mathcal N}
\DeclareMathOperator{\ECW}{\mathcal E \mathcal C \mathcal W}
\DeclareMathOperator{\AC}{\mathcal A \mathcal C}

\DeclareMathOperator{\Var}{{\mathbb V}\mbox{ar}}
\DeclareMathOperator{\Prob}{{\mathbb P}\mbox{r}}
\DeclareMathOperator{\EC}{\mathcal E\mathcal C}
\DeclareMathOperator{\Sph}{\mathcal S}

\DeclareMathOperator*{\circs}{\circ}
\DeclareMathOperator*{\argmax}{arg\,max}

\newcommand{\hSigma}{ \hat \Sigma}

\newcommand{\I}{\imath}%
\renewcommand{\hat}{\widehat}
\renewcommand{\tilde}{\widetilde}

\newcommand{\citep}{\cite}
\newcommand{\citet}{\cite}
\newtheorem{definition}[theorem]{Definition}
\newtheorem{lemma}[theorem]{Lemma}
\begin{document}
\title{Elliptically-Contoured Tensor-variate Distributions with
  Application to Improved Image Learning}
\author{Carlos~Llosa-Vite~and~Ranjan~Maitra
  \thanks{C. Llosa-Vite and R. Maitra are with the Department of Statistics
at Iowa State University, Ames, Iowa 50011, USA. E-mail:
\{cllosa,maitra\}@iastate.edu.}
}




\maketitle
\begin{abstract}
     Statistical analysis of tensor-valued data has largely used the tensor-variate normal (TVN) distribution that may be inadequate when data comes from distributions with heavier or lighter tails. We study a general family of elliptically contoured (EC) tensor-variate distributions and derive its characterizations, moments, marginal and conditional distributions, and the EC Wishart distribution. We describe procedures for maximum likelihood estimation from data that are (1) uncorrelated draws from an EC distribution, (2) from a scale mixture of the TVN distribution, and (3) from an underlying but unknown EC distribution, where we extend Tyler's robust estimator.  A detailed simulation study highlights the benefits of choosing an EC  distribution over the TVN for heavier-tailed data. We develop tensor-variate classification rules using  discriminant  analysis and EC errors and show that they better predict cats and dogs from images in the Animal Faces-HQ dataset than the TVN-based rules. A novel tensor-on-tensor regression and  tensor-variate analysis of variance (TANOVA) framework under EC errors is also demonstrated to better characterize  gender, age and ethnic origin than the usual TVN-based TANOVA in the celebrated Labeled Faces of the Wild dataset. 
   \end{abstract}

\maketitle

\section{Introduction}\label{sec:introduction}
Tensor- or array-variate data arise in many disciplines, for example, in
the context of relational data in political science, imaging
applications in medicine or engineering, artificial
intelligence and so on \citep{bietal21}. 
The tensor-variate normal (TVN)
distribution~\citep{akdemirandgupta11,hoff11,manceuranddutilleul13,ohlsonandetal13} is often used to model and analyze such datasets, as it 
generalizes the matrix-variate normal (MxVN) distribution
\citep{guptaandnagar99,thompsonetal20} to array-variate data. The TVN
and MxVN distributions flow  from  the  multivariate normal (MVN)
distribution, and so 
benefit from the latter's  ease of interpretation, computation and
inference, but also inherit its shortcomings when modeling data that arise from heavy- or light-tailed distributions.

Elliptically contoured (EC)
distributions~\citep{schoenberg38,lord54,kelker70,guptaandetal72,
muirhead82,fangetal90,frahm04,arashi17}
are a flexible class of symmetric vector-variate 
distributions that generalize the MVN distribution, and facilitate the
modeling of data with heavy or light tails. 
In Section~\ref{sec:ecsec} we review and study properties of EC
tensor-variate distributions by characterizing them, their marginal
and conditional distributions and moments. We also introduce the EC
tensor-variate Wishart distribution. Inference is no longer as 
straightforward under EC errors, so Section \ref{sec:MLE} 
provides computationally practical methodology for parameter estimation
under three different scenarios, including a reduced rank tensor-on-tensor
regression (ToTR) and tensor-variate analysis of variance (TANOVA)
framework with EC errors, and a robust Tyler estimator for when data comes from an underlying but unknown EC distribution. We evaluate  
performance of our algorithms in Section~\ref{sec:performance}.
  Section \ref{sec:application} shows
  the value of our methodology in two real-data scenarios.
First, we develop discriminant analysis classification rules using EC tensor-variate
distributions  that exploit the maximum likelihood estimation
frameworks developed in Section \ref{sec:MLE}. We use these classification rules
to compare the predictive performance of the TVN and the tensor-variate-$t$ (TV-$t$)
distributions in classifying cats and dogs from their images in the  Animal
Faces-HQ (AFHQ) database. In all cases, the flexible tensor-variate-t distribution with unknown degrees of freedom outperforms the TVN in terms of area under the receiver operating characteristic and precision-recall curves.
 Our second application demonstrates the ability of our ToTR and TANOVA methodology with EC errors to characterize the Labeled Faces in the Wild (LFW) dataset in terms of
age, ethnic origin and gender, and compares this performance relative to that of TVN-based ToTR and TANOVA. Similarly, assuming EC errors results in preferred models according to the Bayesian information criterion. We conclude this article with a
discussion on our contributions and propose further
generalizations. An online supplement, with 
sections, theorems, equations and figures bearing the prefix ``S'', provides
additional technical details and supporting information as well as proofs of our theorems and lemmas.
\section{Definitions and Characterizations}\label{sec:ecsec}
\subsection{Background and preliminaries}
\label{sec:prelim}
We define a {\em tensor} as a multi-dimensional array of numbers. This
article uses $\mX$, $\sfX$ and  $X$ to denote deterministic tensors, matrices and
vectors, with bold-faced fonts for their random counterparts ({\em
  e.g.}, $\rtX$, $\rmX$ and $\rvX$ denote random tensors, matrices and vectors). Further, we assume that  $\mX \in \mathbb{R}^{\times_{k=1}^p m_k}$  has $(i_1,i_2\sdots,i_p)$th element written
as $\mX(i_1,i_2,\sdots,i_p)$. 
Tensor reshapings \citep{koldaandbader09} allow us to modify the
structure of a tensor while preserving its elements. We can reshape
$\mX$ into a  $(m_{k})\times(\prod_{i=1,i\neq k}^p m_{i})$-matrix
$\mX_{(k)}$ using its $k$th mode matricization.
The matrix $\mX$ can also be reshaped, by means of its $k$th
canonical matricization, into 
a $(\prod_{i=1}^km_{i})\times(\prod_{i=k+1}^pm_{i})$ matrix
$\mX_{<k>}$, and into a vector $\vecc(\mX)$ of  
$\prod_{i=1}^pm_{i}$ elements using vectorization -- see
\eqref{defvec}, \eqref{kmode} and \eqref{canonical} in
Section~\ref{app:tensdef} for formal definitions.

The $k$th mode product between $\mX$ and the $n_k\times m_{k}$ matrix
$\sfA_k$ multiplies $\sfA_k$ to the $k$th mode of $\mX$, resulting in
the tensor  $\mX\times_k \sfA_k\in\mathbb{R}^{m_1\times\sdots\times m_{k-1}\times n_k\times m_{k+1}\times\sdots\times m_p}$. 
Applying the $k$th mode product with respect to $\sfA_k$ to every mode of
$\mX$ results in the Tucker (TK) product $[\![\mX;\sfA_1,\sdots,\sfA_p]\!]\in\mathbb{R}^{\times_{k=1}^pn_k}$ \citep{kolda06}. 
The inner product between two equal-sized tensors $\mX$ and $\mY$ is
defined as $(\vecc\mX)^\top(\vecc\mY)$ and denoted as $\langle\mX,\mY\rangle$. 
 If $\mB$ is a $(p+q)$-way tensor of size $m_1\times m_2\times \ldots \times m_p\times h_1\times h_2\times \ldots \times h_q$, then the partial contraction \citep{llosaandmaitra20} between $\mB$ and $\mX$ (denoted as $\langle \mX|\mB\rangle$) is a $q$-way tensor of size $h_1\times h_2\times \ldots \times h_q$ with 
\begin{equation}\label{eq:partcontr}
\begin{split}
\langle \mX|\mB\rangle(j_1,\dots,j_q)
=
&\sum_{i_1=1}^{m_1}\ldots \sum_{i_p=1}^{m_p}\mX(i_1, \ldots ,i_p)\mB(i_1, \ldots ,i_p,j_1, \ldots ,j_p).
\end{split}
\end{equation}
We refer to Section \ref{app:tensdef} for the reshaping and product definitions encountered in this section.

A random tensor $\rtX$ is a tensor whose vectorized form $\vecc(\rtX)$ is a
random vector. 
In many cases $\Var\{\vecc(\rtX)\}=\sigma^2 \otimes_{k=p}^1\Sigma_k$ for
$\Sigma_k\in\mathbb{R}^{m_k\times m_k}$, and the squared Mahalanobis
distance (with respect to the scale matrices $\sigma^2\Sigma_1,\sdots,\Sigma_p$)  of $\rtX$ from its mean $\mbE(\rtX) = \mM$ is 
\begin{equation}\label{eq:mahlanobis}
D^2_{\sigma^2\Sigma}(\rtX,\mM)=
\dfrac{1}{\sigma^2}
\langle
\rtX - \mM
,
[\![\rtX- \mM;\Sigma_1^{-1},\sdots,\Sigma_p^{-1}]\!]
\rangle.
\end{equation}

A random vector $\bX$ has a vector-valued \textit{spherical}
distribution if its distribution is invariant to rotations: that is, for every $\Gamma\in O(h)= \{H \in 
\mbR^{h\times h} :HH^\top=I_h\}$, we have 
$\bX\overset{d}{=}\Gamma\bX$, or equivalently $\mL(\bX)=\mL(\Gamma\bX)$,
with $\overset{d}{=}$ meaning distributional equality and $\mL(\bX)$ meaning the law of $\bX$.
 This notion of sphericity can be extended to a random matrix $\rmX$
 if it is invariant to rotations of its rows or/and columns
 \citep{fangandchen84}. 	
In this article, we say that a random tensor $\rtX \in \mathbb{R}^{\times_{k=1}^p m_k}$ follows a tensor-valued spherical distribution if $\vecc(\rtX)$ follows a vector-valued spherical distribution. Thus
$
\vecc(\rtX)\overset{d}{=} \Gamma\vecc(\rtX)
$
 for any $\Gamma\in O(\prod_{k=1}^ph_k)$, meaning that 
the distribution of  $\rtX$ is invariant
  under the group of transformations $\mathscr{G} =
 \{\varrho_\Gamma:\Gamma \in O(\prod_{k=1}^ph_k)\}$, where
 $\varrho_\Gamma (\rtX) = \Gamma\vecc(\rtX)$ is the matrix product of
 $\vecc(\rtX)$ with
 the orthogonal matrix $\Gamma$. Since $\langle \rtX,\rtX\rangle$ is
 maximally invariant under $\mathscr{G}$ (cf. Example 2.11 of \citep{eaton89}), it follows
 that the probability density function (PDF) $f(\cdot)$ and
 characteristic function (CF) $\psi(\cdot)$  are of the form
\begin{equation}\label{spher:cf1}
\psi_{\rtX}(\mZ) = \varphi(\langle \mZ,\mZ\rangle),\quad
f_{\rtX}(\mX) = g(\langle\mX,\mX\rangle)
\end{equation}
for some functions $\varphi$ and $g$ called the characteristic
generator (CG) and probability density generator of $\rtX$.
We write $\rtX \sim \Sph_{h_1,\sdots,h_p}(\varphi)$ if $\rtX$ has
the CF in \eqref{spher:cf1}. We define EC-distributed random tensors
from spherical distributions through the TK product. In the following,
we will define EC distributions that are nonsingular, and have positive definite scale matrices.
\begin{definition}\label{def:ellipcon}
Suppose that $\rtX\sim \Sph_{h_1,\sdots,h_p}(\varphi)$, $\mM \in \mathbb{R}^{\times_{k=1}^pm_k}$, and $\sfQ_k\in \mathbb{R}^{m_{k} \times h_k}$ are matrices such that $\sfQ_k \sfQ_k^\top = \Sigma_k$ is positive definite for all $k = 1,\sdots,p$. Then if 
\begin{equation}\label{eq:defscale}
\rtY \overset{d}{=} \mM+[\![\rtX ; \sfQ_1,\sdots,\sfQ_p ]\!],
\end{equation}
we say that $\rtY$ has an EC distribution with mean $\mM$, scale matrices $\Sigma_1,\Sigma_2,\sdots,\Sigma_p$, and CG $\varphi$, written as $\rtY \sim \EC_{\bbm}(\mM,\Sigma_1,\Sigma_2,\sdots,\Sigma_p,\varphi),$ where $\bbm=(m_1,m_2,\sdots,m_p)^\top$.
\end{definition}
For $\varphi(x) = \exp(-x/2)$
,  $\rtY$ follows the TVN distribution $\N_{\bbm}(\mM,\Sigma_1,\sdots,\Sigma_p)$.
We will use
$\bbm\!=\!(m_1,\sdots,m_p),
m=\prod_{i=1}^p m_{i}, 
m_{-k}\!=\!\frac{m}{m_k},
\Sigma \!=\! \otimes_{i=p}^1 \Sigma_i,
\Sigma_{-k}\!=\!\otimes_{i=p,i\neq k}^1 \Sigma_i.
$

Scale mixtures of TVN are an important sub-family of EC distributions, and are defined as 
\begin{equation}\label{eq:scaletensnorm}
\mL(\rtY|(Z = z)) =  \N_{\bbm}
(\mM,z^{-1}\Sigma_1,\Sigma_2,\sdots,\Sigma_p)
\end{equation}
for some non-negative random variable $Z$. This distribution is the tensor-variate extension~\citep{arashi17} of
the vector-variate case studied in
\citet{andrewsandmallows74,chu73,yao73} and the matrix-valued case
investigated in \citet{guptaandnagar99,guptaandvarga95}. 
A commonly-used mixing distribution in  \eqref{eq:scaletensnorm}, for the vector-variate case, uses
$Z\sim Gamma(a/2,b/2)$: doing so in the tensor-variate case yields the PDF
\begin{equation}\label{dengengamma}
f_{\rtY}(\mY) = \Big|\pi b\bigotimes_{k=p}^1\Sigma_k\Big|^{-\frac{1}{2}} 
\dfrac{\Gamma((m+a)/2)}{\Gamma(a/2)}
\Big(1+\dfrac{D^2_{\Sigma}(\mY,\mM)}{b}\Big) ^{-\frac{m+a}{2}}
\end{equation}
that corresponds to the TV-$t$ distribution with
$\nu$ degrees of freedom (DF) when $a=b=\nu$, the tensor-variate Cauchy
distribution when additionally $\nu=1$, and 
a tensor-variate Pearson Type VII distribution (cf. page 450 of {pearson1916})  with parameter $q$ when $a=m$ and $b=q$. We now
derive some useful properties of the EC tensor-variate distribution.

\subsection{Characterizing the EC family of tensor-variate distributions}\label{sec:secproperties}
 Defining EC distributions in terms of
 \eqref{eq:defscale} allows us to extend results from the
 spherical family of distributions to its EC counterpart. For
 instance, from  \eqref{spher:cf1} and \eqref{eq:defscale}, we can write the CF of $\rtY$ in terms of its CG $\varphi$ as
\begin{equation}\label{ellipt:cf1}
\psi_{\rtY}(\mZ) 
=
\exp(\I\langle \mZ , \mM\rangle)\varphi(\langle \mZ , [\![\mZ ; \Sigma_1,\Sigma_2,\sdots,\Sigma_p ]\!]\rangle).
\end{equation}
Similarly, if $\rtX\sim \Sph_{h_1,h_2, \sdots,h_p}(\varphi)$ has the PDF in \eqref{spher:cf1}, then the transformation induced in \eqref{eq:defscale} has Jacobian determinant $|\bigotimes_{k=p}^1\Sigma_k|^{-1/2}$, so the PDF of $\rtY$ is
\begin{equation}\label{ellipt:PDF1}
f_{\rtY}(\mY) = \Big|\bigotimes_{k=p}^1\Sigma_k\Big|^{-1/2} g ( D^2_{\Sigma}(\mY,\mM)),
\end{equation}
where $D^2(\cdot)$ is the squared Mahalanobis distance of \eqref{eq:mahlanobis}. Table~\ref{table:density_generators} 
\begin{table}[b]
\centering
  \caption{Some common EC distributions, defined using 
    \eqref{ellipt:PDF1} with their  corresponding probability 
    density     generators $g(\dot)$, given upto their constant of proportionality.}
\begin{tabular}{ c | c | c }
    \hline
  Distribution & Additional  & $g(x)$ \\ 
     &  parameters &  \\  \hline 
      Normal  & -- &$\exp(-x/2)$\\
   Student's-$t$  & $q>0$  & $(1+q^{-1}x)^{-(q+m)/2}$\\
   Pearson Type VII  & $q>0$  & $(1+x/q)^{-m}$\\
   Kotz Type  & $q>0$  & $x^{m-1}\exp(-qx)$\\
   Logistic  & -- & $\exp(-x)/(1+\exp(-x))^2$\\
   Power exponential  & $q>0$& $\exp(-x^q/2)$\\  \hline 
     \end{tabular}
\label{table:density_generators}
\end{table}
defines some EC distributions by specifying the density generator
function $g(\cdot)$  (for additional specifications, see \citep{arashiandtabatabaey10,galeaandriquelme00}).

The following theorem allows us to express EC tensor-variate distributions in terms of vector-variate EC distributions that have
been studied extensively in the literature.
\begin{theorem}\label{thm:reshap_EC}
  With the reshapings of \eqref{defvec}, \eqref{kmode} and
  \eqref{canonical}, and for all $k=1,2,\sdots,p$, the following
  statements are equivalent, with $n_k = \prod_{i=1}^k m_{i}$:
\begin{enumerate}
\item \label{reshap:par1} $\rtY \sim \EC_{\bbm}(\mM,\Sigma_1,\Sigma_2,\sdots,\Sigma_p,\varphi)$.
\item \label{reshap:par2}
$
\rtY_{(k)} \sim \EC_{m_{k},m_{-k}} ( \mM_{(k)}, \Sigma_k,  \Sigma_{-k} ,\varphi)$
\item \label{reshap:par3}
$
\rtY_{<k>} \sim \EC_{n_k,m/n_k} ( \mM_{<k>},  \bigotimes_{i=k}^1 \Sigma_i, \bigotimes_{i=p}^{k+1} \Sigma_i ,\varphi)$.
\item \label{reshap:par4} $\vecc(\rtY) \sim \EC_{m} ( \vecc(\mM),  \Sigma ,\varphi)$.
\end{enumerate}
\end{theorem}
\begin{proof}
See the Supplementary Material Section \ref{proof:reshap_EC}.
\end{proof}
Spherical distributions receive their name from the observation that if  $\rtX \sim \Sph_{h_1,h_2,\hdots,h_p}(\varphi)$, then $\rtX = R\rtU,$ where the magnitude $R=||\rtX||$ is independent from $\rtU$
and $\vecc(\rtU)$ is uniformly distributed in the shell of a unit sphere of dimension $\prod_{k=1}^ph_k$ \citep{cambanisetal81,schoenberg38,khokhlov06}. 
Based on the CF in \eqref{spher:cf1} and the independence of $R$ and $\rtU$ , we have that
$
\psi_{\rtX}(\mZ) 
= 
\mbE\{\psi_{\rtU}(R\mZ)\}, 
$
and since $\rtU$ is also spherically distributed, the CG of $\rtX$ can
be written, in terms of the CG of $\rtU$, as
\begin{equation}\label{spher:stocharacgen}
\varphi(u) =\mbE\{\varphi_{\rtU}(R^2u)\}.
\end{equation}
The distribution of $R$ determines the distribution of $\rtX$. For instance, if $R^2\sim\chi^2_h$, where $h=\prod_{k=1}^ph_k$, then 
 $\varphi(u) =
 \exp(-u/2)$ and $\rtX$ follows a TVN distribution.
Further, if
$\rtY\sim\EC_{\bbm}(0,\sigma^2\Sigma_1,\sdots,\Sigma_p,\varphi)$ as in equation \eqref{eq:defscale} and
$\rtX = R\rtU$ as above,
 then $\rtZ
= \rtY/ ||\rtY||$ does not depend on $R$, meaning that the
distribution of $\rtZ$ does not depend on the original
EC distribution of $\rtY$. The distribution of $\rvZ=\vecc(\rtZ)$  is called
the \textit{elliptical angular distribution}
\citep{mardia72,watson83}, and is denoted as $\rvZ\sim\AC_m(\Sigma)$, with PDF 
\begin{equation}\label{eq:angular_pdf}
f_{\rvZ}(\bz) = \frac{\Gamma(m/2)}{2\pi^{m/2}}|\Sigma|^{-1/2}(\bz^\top\Sigma^{-1}\bz)^{-m/2}.
\end{equation}
The above PDF does not depend on $\sigma^2$. In Section
\ref{sec:robustmle}, we derive a robust tensor-variate Tyler
estimator that exploits the fact that $\rtZ$ is the same for any
EC distribution. 

\subsection{Marginal and conditional distributions}
Our next theorem shows that the TK product of an EC distributed
random tensor is also EC distributed. 
\begin{theorem}\label{ellipt:distform} 
Let $\rtY\! \sim\!  \EC_{\bbm}(
\mM,\Sigma_1,\sdots, \Sigma_p, \varphi)$ and $\sfA_k \in \mathbb{R}^{n_k\times m_{k}}$  $\forall k=1,\sdots,p$. Then 
\begin{equation*}
\begin{split}
[\![\rtY; \sfA_1,\sdots,\sfA_p ]\!] \sim \EC_{[n_1 , \sdots , n_p]}&( [\![\mM; \sfA_1,\sdots,\sfA_p ]\!],\sfA_1\Sigma_1\sfA_1^\top,\sdots, \sfA_p\Sigma_p\sfA_p^\top,\varphi).
\end{split}
\end{equation*}
\end{theorem}
\begin{proof}  
See the Supplementary Material Section \ref{proof:ellipt:distform}.
\end{proof}
As a corollary, the marginal distributions follow from Theorem \ref{ellipt:distform} by choosing $\sfA_1,\sdots,\sfA_p$ appropriately, as described in detail in Section \ref{sec:marginals}. 
We now extend the vector-variate result of \citet{cambanisetal81} to derive conditional distributions in the tensor-variate setting.
\begin{theorem}\label{ellipt:conditional}
Suppose that $\rtY \sim \EC_{\bbm}(\mM,\Sigma_1,\sdots,\Sigma_p
,\varphi)$, where $m_p = n_1 + n_2>1$ for some $n_1, n_2 \in
\mathbb{N}$, and $\Sigma_p$ is a block matrix with $(i,j)$th block $\Sigma_{ij} \in \mathbb{R}^{n_i\times n_j}$.
Partition $\rtY$ and $\mM$ over the $pth$ mode with 
$\rtY_1,\mM_1 \in \mathbb{R}^{m_1\times m_2\times\hdots\times m_{p-1}\times
  n_1}$ and $\rtY_2,\mM_2 \in \mathbb{R}^{m_1\times m_2\times \hdots\times
  m_{p-1}\times n_2}$. Then 
\begin{equation}\label{eq:condres}
\begin{split}
\rtY_1|(\rtY_2 = \mY_2)\sim \EC_{[m_1,\sdots,m_{p-1},n_1]}
\Big(
\mM_{\{1\}|\{2\}}
,
\Sigma_1,\sdots,\Sigma_{p-1},
\Sigma_{p,\{11\}|\bullet}
,\varphi_{q(\mY_2)}
\Big),
\end{split}
\end{equation}
where
$\mM_{\{1\}|\{2\}} = \mM_1 + (\mY_2 - \mM_2)\times_p(\Sigma_{p,12}\Sigma_{p,22}^{-1})$ 
, 
$\Sigma_{p,\{11\}|\bullet} = \Sigma_{p,11}-\Sigma_{p,12}\Sigma_{p,22}^{-1}\Sigma_{p,21}$
,
$\varphi_{q(\mY_2)}(u) =\mbE\{R^2_{q(\rtY_2)}\varphi_{\rtU}(u)|\rtY_2 = \mY_2\}$,
$
R^2_{q(\rtY_2)} \!=\! D^2_{\Sigma_{\bullet}}(\rtY_1,\mM_1),
$
and
$
\Sigma_{\bullet} \!=\! \Sigma_{p,\{11\}|\bullet}\otimes \Sigma_{-1}.
$
\end{theorem}
\begin{proof}
See the Supplementary Material Section \ref{proof:ellipt:conditional}.
\end{proof}
The CG  $\varphi_q(\mY_2)$ in  \eqref{eq:condres} is a conditional moment, just like in the more general case of \eqref{spher:stocharacgen}. Although Theorem \ref{ellipt:conditional} applies only for the conditional distribution that results from partitioning the last mode of the random tensor, we can find conditional distributions of any subtensor by applying Theorem \ref{ellipt:conditional} multiple times, as demonstrated in Section \ref{supp:condsupp}.

\subsection{Moments}

  The moments of EC tensor-variate distributions are found by
  differentiating \eqref{ellipt:cf1}. We provide the first four
  moments in Theorem \ref{expectelement} and then use
  them in Theorem \ref{thm:moments_ec} to find moments of other special forms. 
\begin{theorem}\label{expectelement} 
Suppose $\rtY \sim  \EC_{\bbm}( \mM,\Sigma_1,\Sigma_2,\sdots,
\Sigma_p,\varphi)$, and let $\bbi=(i_1,\sdots, i_p)^\top$,
$\bj=(j_1,\sdots,j_p)^\top$, $\bk=(k_1,\sdots,k_p)^\top$ and $\bl=(l_1,\sdots,l_p)^\top$ be sets of indices such that $i_q,j_q,k_q,l_q\in \{1,\sdots,m_q\}$ for $q =1,\sdots,p$. Further, denote $\rtY(i_1,\sdots, i_p)=Y_{\bbi} $, $\mM(i_1,\sdots, i_p)=m_{\bbi}$ and $\sigma_{\bbi\bj} = \prod_{q=1}^p \Sigma_q(i_q,j_q)$. Then
\begin{enumerate}
\item \label{expectelement:1} $\mbE(Y_{\bbi}) = m_{\bbi}$.
\item  \label{expectelement:2}
$\mbE(Y_{\bbi}Y_{\bj}) 
= m_{\bbi}m_{\bj}-2 \varphi'(0)\sigma_{\bbi\bj}$.
\item \label{expectelement:3} $\mbE(Y_{\bbi}Y_{\bj}Y_{\bk}) 
\hspace*{-.1cm}=\hspace*{-.1cm}
 m_{\bbi}m_{\bj}m_{\bk}
 \hspace*{-.1cm}-\hspace*{-.1cm}
 2\varphi'(0)(
m_{\bbi}\sigma_{\bk\bj}
\hspace*{-.1cm}+\hspace*{-.1cm}
m_{\bj} \sigma_{\bbi\bk}
\hspace*{-.1cm}+\hspace*{-.1cm}
m_{\bk}\sigma_{\bbi\bj}
).
$
\item \label{expectelement:4} $\mbE(Y_{\bbi}Y_{\bj}Y_{\bk}Y_{\bl}) = m_{\bbi}m_{\bj}m_{\bk}m_{\bl}
+ 4 \varphi''(0)\Big(\sigma_{\bk\bl}\sigma_{\bj\bbi}
+\sigma_{\bj\bl}\sigma_{\bbi\bk}
+\sigma_{\bbi\bl}\sigma_{\bj\bk}\Big)-2 \varphi'(0)\Big(m_{\bl}m_{\bk}\sigma_{\bbi\bj}
+m_{\bj}m_{\bl}\sigma_{\bbi\bk}
+m_{\bj}m_{\bk}\sigma_{\bbi\bl}
+m_{\bbi}m_{\bl}\sigma_{\bj\bk}
+m_{\bbi}m_{\bk}\sigma_{\bj\bl}
+m_{\bbi}m_{\bj}\sigma_{\bk\bl}\Big).
$
\end{enumerate}
Here, the $k$th statement of the theorem requires $\varphi^{(k)}(0)<\infty$, for $k=1,2,3,4$.
\end{theorem}
\begin{proof}
See  the Supplementary Material Section \ref{app:proofexpectelement}.
\end{proof}
Theorem \ref{ellipt:distform} allows us to derive many different
moments. We derive some of them now.
\begin{theorem}\label{thm:moments_ec}
Let $\rtY \sim \EC_{\bbm}(\mM,\Sigma_1,\Sigma_2,\hdots,\Sigma_p,\varphi)$, where $\varphi''(0)\!<\!\infty$. 
Further, let $\sfA_k$ and $\sfB_k$ be $n_k\times m_k$ and $h_k\times
m_k$ matrices for $k=1,2,\hdots,p$, and define $D^2_A(\mX) = \langle
\mX,[\![\mX;\sfA_1,\sfA_2,\sdots,\sfA_p]\!]\rangle$ whenever each
$A_k$ is a square matrix. Then,
\begin{enumerate}
\item \label{par51} $\mbE(\rtY) = \mM$.
\item \label{par52}
$
\Var\{\vecc(\rtY)\} = -2\varphi'(0)\bigotimes\limits_{k=p}^1 \Sigma_k.
$
\item \label{par54} If $n_k=m_k$ for all $k=1,2,\hdots,p$, then 

$
\mbE\left\{D^2_A(\rtY)\right\} = D^2_A(\mM) -2\varphi'(0)\prod\limits_{k=1}^p \tr(\Sigma_k\sfA_k^\top).
$
\item\label{par55} If $\mV$ is of size $n_1\times n_2\times\hdots\times n_p$, then 

$\mbE(\langle \mV,[\![\rtY;\sfA_1,\hdots,\sfA_p]\!]\rangle\rtY) 
=\langle \mV,[\![\mM;\sfA_1,\hdots,\sfA_p]\!]\rangle\mM
 -
2\varphi'(0)[\![\mV;\Sigma_1 \sfA_1^\top,\hdots,\Sigma_p
\sfA_p^\top]\!].
$
\item\label{par56} If $n_k=m_k$ for all $k=1,2,\hdots,p$, then

$ -2\varphi'(0)\Big[\prod\limits_{k=1}^p \tr(\Sigma_k\sfA_k^\top)\mM 
 +[\![\mM;\Sigma_1 \sfA_1,\hdots,\Sigma_p \sfA_p]\!] 
   	+
 [\![\mM;\Sigma_1 \sfA_1^\top,\hdots,\Sigma_p \sf A_p^\top]\!]\Big].
$
\item \label{par57} If $\mM=0$, $\varphi^{(4)}(0)<\infty$ and $n_k=h_k = m_k$ for all $k=1,\hdots,p$, then 

$\mbE\big
  \{
D^2_A(\rtY)D^2_B(\rtY)
\big\} \!=\!
 4\varphi''(0)\Big[
\prod\limits_{k=1}^p \left\{ \tr(\sfA_k\Sigma_k)\tr(\sfB_k\Sigma_k) \right\}
+
\prod\limits_{k=1}^p \tr(\sfA_k\Sigma_k\sfB_k^\top\Sigma_k) 
\!+\!
\prod\limits_{k=1}^p \tr(\sfA_k\Sigma_k\sfB_k\Sigma_k) 
\Big].
$
\end{enumerate}
\end{theorem}
\begin{proof}
See the Supplementary Material Section \ref{proof.b2}.
\end{proof}
From Parts~\ref{par54} and \ref{par57} of Theorem
\ref{thm:moments_ec}, we see that for $\rtY$ specified as in Definition \ref{def:ellipcon},  $\mathbb{E}\{ D^2_{\Sigma}(\rtY,\mM)\}= -2\varphi'(0)\times m
$, and
$
\mathbb{V}ar\{ D^2_{\Sigma}(\rtY,\mM)\}=4\times [\varphi''(0)\times(m^2 + 2m) - (\varphi'(0)\times m)^2]. 
$

\subsection{The EC tensor-variate Wishart distribution} The Wishart
distribution~\citep{wishart28} plays an important role in multivariate
statistics in inference on sample dispersion
matrices \citep{johnsonandwichern08}.  The matrix-variate EC Wishart
distribution was proposed
in~\citet{caroetal14,diazandgutierrez11,sutradharandali89} as per 
\begin{definition}\label{def:matwish}
A matrix $\rmS$ is said to follow an EC matrix-variate  Wishart
distribution with DF  $n$, scale parameter $\Sigma$ and
CG $\varphi$, written as $\rmS\sim \ECW_m(n,\Sigma,\varphi)$, if there exists a random matrix  $\rmX\sim \EC_{m,n}(0,\Sigma,I_n,\varphi)$ such that $\rmS \overset{d}{=} \rmX\rmX^\top$.
\end{definition}
If $n>m$, then the PDF of $\rmS$ can be written as 
\begin{equation}
f_{\rmS}(\sfS)=
\dfrac{\pi^{mn/2}}{\Gamma_m(n/2)|\Sigma|^{n/2}}|\sfS|^{(n-m-1)/2}g(\tr(\Sigma^{-1}\sfS)),
\end{equation}
where $g(\cdot)$ is the density generator of $\rmX$. For $n<m$, we have the
singular matrix-variate EC Wishart distribution. For completeness, we derive the PDF of the functionally independent elements of $\rmS$  when $n<m$ in  Section \ref{sec:singularecwishart}.
 We now define the EC TV Wishart distribution. 
\begin{definition}\label{def:tenswish}
The random tensor 
$\rtS \in\mathbb{R}^{\times_{j=1}^p m_j}
$ 
 follows a tensor-valued
EC Wishart distribution with DF $n$, scale matrices
$\Sigma_1,\Sigma_2,\sdots,\Sigma_p$ and CG $\varphi$ ({\em i.e.},
$\rtS\sim
\ECW_{\bbm}(n,\Sigma_1,\sdots,\newline\Sigma_p,\varphi)$)
if $\rtS_{<p>}\sim \ECW_{m}(n,\Sigma,\varphi)$. 
\end{definition}
In the vector-variate case, the Wishart distribution can be
constructed from the vector outer product sum of independent and identically distributed (IID) MVN-distributed
random vectors. Similarly, the EC tensor-variate Wishart distribution can be constructed from the outer tensor product sum of uncorrelated tensors $\rtY_1,\rtY_2,\sdots,\rtY_n$ (where $\rtY_i = \rtY\times_{p+1}{\be_i^n}^\top$ for all $i=1,2,\dots,n$ and $\be_i^n \in \mbR^n$ is a unit basis vector with 1 as the $i$th
element and 0 elsewhere) which jointly follow
\begin{equation}\label{eq:jointfull}
\rtY\sim\EC_{[\bbm,n]^\top} (\mM_n, \sigma^2\Sigma_1,\Sigma_2, \sdots, \Sigma_p,I_n,\varphi).
\end{equation}
Here $\Sigma_k$ is constrained to have $\Sigma_k(1,1) = 1$ for all $k=1,2,\dots,p$ while $\sigma^2$ captures the overall proportionality constraint of $\Sigma$, and 
 $\mM_{n}\times_{p+1}{\be_i^n}^\top = \mM$ for all $i=1,2,\sdots,n$.
 In the next theorem, we show that $\rtY_1 ,\rtY_2,\sdots\rtY_n$ are uncorrelated, have identical marginal distributions,  and that their tensor outer product sum follows the EC Wishart distribution.
\begin{theorem}\label{thm:tenswish1}
Let $ \rtY_1, \rtY_2\hdots, \rtY_n$ jointly follow the distribution in
\eqref{eq:jointfull}. Then, we have:
\begin{enumerate}
\item \label{thm:tenswish1:1}
 $\rtY_i\sim\EC_{\bbm} (\mM, \sigma^2\Sigma_1, \Sigma_2,\sdots, \Sigma_p,\varphi)$ and $\forall i\neq j, \mathbb{E}\left(\vecc(\rtY_i-\mM)\vecc(\rtY_j-\mM)^\top\right) = 0$.
\item \label{thm:tenswish1:2}
$
\sum\limits_{i=1}^n (\rtY_i \circ \rtY_i)
 \sim 
\ECW_{\bbm}(n,\sigma^2\Sigma_1,\Sigma_2,\sdots,\Sigma_p,\varphi).
$
\end{enumerate}
Here $\circ$ denotes the tensor outer product that we define in Section \ref{app:tensdef}.
\end{theorem}
\begin{proof}
See the Supplementary Material Section \ref{proof:tenswish1}.
\end{proof}
In Section \ref{sec:uncorrML} we derive maximum likelihood estimators (MLEs)
estimators for $(\mM,\sigma^2,\newline\Sigma_1,\Sigma_2,\sdots,\Sigma_p)$
under the  assumption that $\rtY_1,\rtY_2,\sdots,\rtY_n$ are
uncorrelated, as implied in  \eqref{eq:jointfull}, which is a
generalization of the traditional IID assumption.
 Setting $\varphi(z) = \exp(-z/2)$ leads to 
the case where IID $\rtY_i\sim\N_{\bbm}(0,\sigma^2\Sigma_1,\Sigma_2,\sdots,\Sigma_p)$ results in $\sum_{i=1}^n (\rtY_i \circ \rtY_i) \sim \mW_{\bbm}(n,\sigma^2\Sigma_1,\Sigma_2,\sdots,\Sigma_p)$.

\section{Maximum likelihood estimation}\label{sec:MLE}
This section derives MLEs from EC tensor-variate distributed data under four
different scenarios.  We first derive these estimators from
uncorrelated draws that have common and joint 
EC tensor-variate distributions. This procedure depends on
identifiable MLEs from the TVN distribution, and so we also
describe how to obtain such estimates.
Next, we provide  expectation conditional maximization (ECM)
\citep{dempsterandetal77,mengandrubin93}  and  ECM Either
(ECME) \citep{liuandrubin94} algorithms for maximum likelihood estimation on
IID realizations from a scale TVN mixture distribution, and show how
these estimation algorithms can be used to extend the ToTR framework
of \citep{llosaandmaitra20} to the case where the error distribution
is a scale mixture of TVNs.
 Finally, we derive a robust estimation algorithm for the scale matrices in the
spirit of Tyler's estimation \citep{tyler83,tyler87a,tyler87b} for
unknown EC distributions.

\subsection{Estimation in the TVN model}\label{sec:tensnormalMLE}

Suppose that we independently observe $\mY_1,\mY_2,\dots,\mY_n$ from $\N_{\bbm}(0,\sigma^2\Sigma_1,\Sigma_2,\dots,\Sigma_p)$. The MLEs of the scale matrices $\Sigma_1,\Sigma_2,\dots,\Sigma_p$  are not identifiable since $\Sigma_k \otimes \Sigma_l =
(a\Sigma_k) \otimes (\Sigma_l/a)$ for any scalar $a\neq 0$. Moreover,
the scale matrix estimators have no closed-form solutions, but each $\hSigma_k$ depends on the rest as $\hSigma_k =\sfS_k/(nm_{-k}\sigma^2)$, where
$
\sfS_k  = \sum_{i=1}^n \mY_{i(k)}
\Sigma_{-k}^{-1}\mY_{i(k)}',$ with $\Sigma_{-k}$ as defined in the
paragraph following Definition~\ref{def:ellipcon}.
The existing methodology \citep{dutilleul99,srivastavaandetal08,dutilleul18,rosandetal16,soloveychikandtryshin16}
takes into account the indeterminacy of the estimates by scaling each
$\hSigma_k$ such that the trace, determinant or $(1,1)$th element is
unity. Alternatively, we use the ADJUST procedure of \citet{glanzandcarvalho18} that optimizes the loglikelihood with respect to $\Sigma_k$ under the constraint $\Sigma_k(1,1)=1$ as
$
\hSigma_k = ADJUST(nm_{-k},\sigma^2,\sfS_k).
$
Here $\sigma^2$ is inexpensively estimated to be 
$\hat\sigma^2=\tr(S_k\hSigma_k^{-1})/(nm)$. These form our minor modifications to
the maximum likelihood estimation algorithm of \citet{akdemirandgupta11,hoff11,manceuranddutilleul13,ohlsonandetal13}, but will be used in the following section.

\subsection{Estimation from uncorrelated data with identical
  EC-distributed marginal distributions}\label{sec:uncorrML}
Suppose that we observe realizations $\mY_1,\mY_2,\sdots,\mY_n$ of random
tensors that have the same setting as Theorem \ref{thm:tenswish1}. Our
derivation of MLEs  under the joint EC distributional
assumption of \eqref{eq:jointfull} is then a tensor-variate extension
of Theorem 1 of \citet{andersonandetal86}, and is as follows:
\begin{theorem}\label{mluncorrelated}
Suppose that the density generator $g(\cdot)$ of $\rtY$ is such that 
$
h(d) = 
d^{nm/2}g(d) 
$
has a finite positive maximum $d_g$. Also, let $(\tilde{\sigma}^2,\tilde{\mM},
\tilde\Sigma_1,\tilde\Sigma_2,\dots,\tilde\Sigma_p)$
be the TVN MLEs described in Section \ref{sec:tensnormalMLE}. Then the MLE for $(\sigma^2,\mM,\Sigma_1,\Sigma_2,\dots,\Sigma_p)$ is\newline $(\frac{nm}{d_g}\tilde{\sigma}^2,\tilde{\mM},\tilde\Sigma_1,\tilde\Sigma_2,\dots,\tilde\Sigma_p)$.
\end{theorem}
\begin{proof}
See the Supplementary Material Section \ref{proof:mluncorrelated}.
\end{proof}

It follows that $d_g = mn$ in the TVN case, $d_g = nmb/a$ when $g(\cdot)$ is
of the form \eqref{dengengamma},  $d_g = (mn/s)^{1/s}$ in the power exponential case, and $d_g$ is the solution to 
$
mn/(2d)  = \tanh\left(d/2\right)
$
in the logistic case~\citep{galeaandriquelme00,fangetal90}.
\subsection{Estimation from TVN scale mixture realizations}\label{sec:ECtensMLE}

Suppose that we observe $n$ IID realizations $\mY_1,\mY_2,\dots,\mY_n$
from the scale mixture distribution of \eqref{eq:scaletensnorm}. The complete-data loglikelihood is
\begin{equation}\label{eq:completelog}
\ell_c(\boldsymbol{\theta}) = -\dfrac{n}{2} \log|\sigma^2\Sigma| -\dfrac{1}{2\sigma^2}\sum_{i=1}^n z_i D^2_{\Sigma}(\mY_i,\mM) 
\end{equation}
where $\boldsymbol{\theta}  
= \left((\vecc\mM)^\top,\boldsymbol{\theta}_{\Sigma}^\top \right)^\top$,
and
$ \boldsymbol{\theta}_{\Sigma}
=
 \big(\sigma^2,(\vecc\Sigma_1)^\top,(\vecc\Sigma_2)^\top\dots,$ $ (\vecc\Sigma_p)^\top\big)^\top.$
 We now use $\ell_c$ to propose an EM algorithm with the following expectation
 (E) and conditional maximization (CM) steps:
\newline
\textbf{E-step:} Here, we obtain the conditional
expectation of \eqref{eq:completelog}, given the  observed data and
evaluated at the $t$th parameter iteration
$\boldsymbol{\theta}^{(t)}$, to get
\begin{equation}\label{eq:expectedcompletelog1}
Q(\boldsymbol{\theta};\boldsymbol{\theta}^{(t)}) = -\dfrac{n}{2} \log|\sigma^2\Sigma|  -\dfrac{1}{2\sigma^2}\sum_{i=1}^n \hat{z}_i^{(t)} D^2_{\Sigma}(\mY_i,\mM),
\end{equation}
where $\hat{z}_i^{(t)}= \mathbb{E}(Z_i|\rtY_i=\mY_i;\boldsymbol{\theta}^{(t)}).$
This conditional expectation depends on the distribution of $Z_i$.
For $Z_i\sim Gamma(a/2,b/2)$, $\rtY_i$ has PDF  given by
\eqref{dengengamma} and  then $\hat{z}_i^{(t)} = {(m+a)}/(D^{2}_{\sigma^{2(t)}\Sigma^{(t)}}(\mY_i,\mM^{(t)})+b).$
\newline
\textbf{CM-steps:} The CM steps maximize 
\eqref{eq:expectedcompletelog1}, but with respect to subsets of parameters. The
first CM step maximizes~\eqref{eq:expectedcompletelog1} with respect to $\mM$ and is obtained by setting to zero
the total differential on $\mM$, that is, 
$\partial Q(\mM) = (\sigma^2)^{-1}\langle \partial \mM, [\![ 
\mS ;
\Sigma_1^{-1},\Sigma_2^{-1},\dots,\Sigma_p^{-1} ]\!] \rangle
 = 0,$ where $\mS =  \sum_{i=1}^n
 \hat{z}_i^{(t)}\mY_i-\mM\sum_{i=1}^n \hat{z}_i^{(t)}.$ The term
 $\partial Q(\mM)=0$ implies that  $\mS =0$, or  
$
\mhM^{(t+1)}  =  (\sum_{i=1}^n \hat{z}_{i}^{(t)} \mY_i) / (\sum_{i=1}^n  \hat{z}_{i}^{(t)}).
$
 Setting  $\mM = \mhM^{(t+1)}$ in  \eqref{eq:expectedcompletelog1}
 profiles out  $\mM$ and again yields \eqref{eq:expectedcompletelog1},
 but with $\mY_{w,i}^{(t)}  = \sqrt{\hat{z}_{i}}^{(t)}(\mY_i -
 \mhM^{(t+1)})$ instead of $\mY_i$.  This profiled loglikelihood is
 identical to the loglikelihood of the TVN model when
 $\mY_{w,i}^{(t)}\sim
 \N_{\bbm}(0,\sigma^2\Sigma_1,\Sigma_2,\dots,\Sigma_p)$, and thus the
 next $(p+1)$ CM steps (corresponding to
 $\sigma^2,\Sigma_1,\Sigma_2,\dots,\Sigma_p$) are obtained similarly
 as for the TVN model of Section \ref{sec:tensnormalMLE}. Specifically, $
\hSigma_k = ADJUST(nm_{-k},\sigma^2,\sfS_k),
$
and $\sigma^2=\tr(S_k\hSigma_k^{-1})/(nm)$. 
\newline
\textbf{EM algorithm:} The ECM algorithm iterates the E- and CM-steps until convergence. We determine convergence based on the relative difference in $||\mM|| + ||\sigma^2\Sigma||,$ where $||\sigma^2\Sigma|| = \sigma^2\prod_{k=1}^p||\Sigma_k||$. The resulting ECM algorithm is summarized in the following steps:
\begin{enumerate}
\item \textbf{Initialization:} Fit the faster TVN model of Section
  \ref{sec:tensnormalMLE} with lax convergence, set $t=0$.
\item 
\textbf{E-Step:} Find $\hat{z}_i^{(t)} := \mathbb{E}(Z_i|\rtY_i=\mY_i;\boldsymbol{\theta}^{(t)})$ $ \forall i=1,2,\dots,n$.
\item 
\textbf{CM-Step 1:} Find $\mhM^{(t+1)}  =  (\sum_{i=1}^n \hat{z}_{i}^{(t)} \mY_i) / (\sum_{i=1}^n  \hat{z}_{i}^{(t)}).$
\item 
\textbf{CM-Step 2:} For each $k=1,2,\dots p$, estimate $\hSigma_k$.
\item
\textbf{CM-Step 3:} estimate $\sigma^2$.
\item
\textbf{Iterate or stop:} Return to the E-step and set $t\gets t+1$, or stop if convergence is met.
\end{enumerate}
\textbf{Alternative EM approach for estimating $\boldsymbol{\sigma^2}$:} 
The above algorithm is an ECM algorithm, since it estimates $\sigma^2$ once for each iteration. A faster alternative is to estimate $\sigma^2$ right after each $\Sigma_k$, making our algorithm an AECM algorithm \citep{mengandvandyk97}. An even faster alternative estimates $\sigma^2$ directly from the loglikelihood 
$\ell(\sigma^2)=
-mn\log(\sigma^2)/2
+\sum_{i=1}^n\log  g (D^{2}_{\Sigma^{(t)}}(\mY_i,\mM^{(t)})/\sigma^2)
$
making the algorithm an ECME algorithm \citep{liuandrubin94}. Our
experiments showed our ECME algorithm to be considerably more
computationally efficient for heavy-tailed EC distributions.

For some cases, such as for the DF parameter in the $t$ distribution, the mixing random variable $Z$ in  \eqref{eq:scaletensnorm} has extra parameters that show up in the last term of the complete loglikelihood \eqref{eq:completelog}, and that may be estimated in  additional
CM steps \citep{liuandrubin95,mcLachlanandkrishnan08}. These extra parameters can also be optimized directly from the loglikelihood as an Either step \citep{liuandrubin94}. For the tensor-variate-$t$ (TV-$t$) distribution with an unknown $\nu$ degree of freedom (the PDF of the TV-$t$ distribution is provided in  \eqref{dengengamma} whenever $a=b=\nu$), we estimate $\nu$ as the solution to
\begin{equation}\label{eq:eithernu}
\begin{split}
0=&1+
\dfrac{1}{n}\sum_{j=1}^n\big[\log z_j^{(k+1)}(\nu)- 
  z_j^{(k+1)}(\nu) \big] 
-
\psi(\dfrac{\nu}{2}) + \psi(\dfrac{n+m}{2}) 
+\log(\dfrac{\nu}{2}) - \log(\dfrac{n+m}{2})
\end{split}
\end{equation}
 at each iteration, where
 $
 z_j^{(k+1)}(\nu) = (\nu+m)/(\nu+D^{2}_{\hat\sigma^{2(t)}\hSigma^{(t)}}(\mY_i,\mhM^{(t)})),
 $
 and $\psi$ is the digamma function.
 This is an adaptation of the vector-variate case of \citet{liuandrubin95} to our tensor-variate case.
\subsection{ToTR with scale TVN mixture errors}\label{sec:totr_ec}
 The algorithm proposed in Section \ref{sec:ECtensMLE} can be extended
 to the case where $(\mY_1,\mY_2,\dots,\mY_n)$ are independent
 realizations of the ToTR model, for $i=1,2,\ldots,n$,
\begin{equation}\label{eq:totr}
\mY_i = \langle \mX_i | \mB \rangle+ \mE_i,
\end{equation}
where  $\mX_i$ is the $i$th tensor-variate regression covariate, 
$\mE_i$  the $i$th regression random error following the scale mixture
distribution of \eqref{eq:scaletensnorm} for $\mM = 0$, $\mB$ is the
tensor-variate regression coefficient,  and $\langle . | . \rangle$
denotes the partial contraction as per \eqref{eq:partcontr}. We assume that $\mX_i\in \mathbb{R}^{h_1\times h_2\times \dots \times h_l}$ and $\mY_i \in \mathbb{R}^{m_1\times m_2\times \dots \times m_p}$, which means that $\mB \in\mathbb{R}^{h_1\times h_2\times \dots \times h_l\times m_1\times m_2\times \dots \times m_p }$. 
ToTR was proposed by \citet{hoff14} under TVN errors and outer matrix product (OP) factorization of $\mB$, by \citet{lock17} under {\em canonical polyadic} or CANDECOMP/PARAFAC
(CP) ~\citep{carrollandchang70,harshman70} factorization of $\mB$, by \citep{gahrooeietal21} under the TK factorization, and
by   \citet{liuetal20} under the TT format~\citep{oseledets11} that is the
tensor-ring (TR) format \citep{zhaoetal16a} when one of the TR ranks
is set to 1. The CP, TR, TK and OP factorizations on $\mB$ allow for substantial parameter
reduction without affecting prediction or discrimination ability of the
regression model. For more recent developments in ToTR see \citep{luoandzhang22,raskuttietal19}. The more general case of ToTR can be found under the CP, OP, TR and TK formats and TVN errors in \citet{llosaandmaitra20}, and here we extend these results to the case where the errors follow a scale mixture of TVN distribution.

\subsubsection*{Maximum Likelihood Estimation}
Following the same steps as in Section \ref{sec:ECtensMLE}, the E-step is performed by calculating the expected complete loglikelihood, which can be written as
\begin{equation}\label{eq:expectedcompletelog_totr}
Q(\boldsymbol{\theta};\boldsymbol{\theta}^{(t)}) = -\dfrac{n}{2} \log|\sigma^2\Sigma|  -\dfrac{1}{2\sigma^2}\sum_{i=1}^n D^2_{\Sigma}(\mY_{w,i}^{(t)},\langle \mX_{w,i}^{(t)} | \mB \rangle)
\end{equation}
with $\boldsymbol{\theta} = (\mB,\sigma^2\Sigma)$, and $(\mY_{w,i}^{(t)},\mX_{w,i}^{(t)})=(\sqrt{\hat{z}_i}^{(t)}\mY_i,\sqrt{\hat{z}_i}^{(t)}\mX_i)$ for $\hat{z}_i^{(t)}= \mathbb{E}(Z_i|\rtY_i=\mY_i;\boldsymbol{\theta}^{(t)}).$
For $Z_i\sim Gamma(a/2,b/2)$, we have 
$
\hat{z}_i^{(t)} = {(m+a)}/(D^{2}_{\sigma^{2(t)}\Sigma^{(t)}}(\mY_i,\langle \mX_i|\mB^{(t)} \rangle)+b).
$
 The CM-steps are performed by sequentially optimizing
 \eqref{eq:expectedcompletelog_totr} with respect to the parameter
 factors involved in $\boldsymbol{\theta}$. This equation
 is identical to the loglikelihood for the ToTR model under TVN
 errors (see Eqn. (16) of \citep{llosaandmaitra20}), and so the CM-steps
 are performed in one iteration of their block-relaxation algorithm,
 but with $\{(\mY_{w,i}^{(t)},\mX_{w,i}^{(t)}),i=1,2,\dots,n\}$ as the
 responses and covariates. This approach is very general as it works
 for any of the low-rank formats studied in
 \citet{llosaandmaitra20}. The ECM algorithm for maximum likelihood estimation is
 summarized as follows: 
\begin{enumerate}
\item \textbf{Initialization:} Fit the respective ToTR model with TVN errors with a lax convergence. 
\item 
\textbf{E-Step:} Find $\hat{z}_i^{(t)} := \mathbb{E}(Z_i|\rtY_i=\mY_i;\boldsymbol{\theta}^{(t)})$ and $(\mY_{w,i}^{(t)},\mX_{w,i}^{(t)})$ $ \forall i=1,2,\dots,n$.
\item 
\textbf{CM-Step:} Perform one iteration of the respective ToTR algorithm of \citet{llosaandmaitra20}, but with $\{(\mY_{w,i}^{(t)},\mX_{w,i}^{(t)}),i=1,2,\dots,n\}$ as the responses and covariates.
\item
\textbf{Iterate or stop:} Return to the E-step and set $t\gets t+1$, or stop if convergence is met.
\end{enumerate}
Section \ref{sec:classification} later demonstrates that the proposed ToTR
framework with EC errors outperforms ToTR with TVN errors in terms of
classification performance on color bitmap images of cats and dogs
from the AFHQ database.

\subsection{Robust tensor-variate Tyler's estimator}\label{sec:robustmle}
While we proposed several methods for maximum likelihood estimation, they all
require the form of the underlying EC distribution to be known. In this section, we derive constrained MLEs that are robust to the type of EC distribution in the spirit of Tyler's estimator \citep{tyler83,tyler87a,tyler87b,soloveychikandtryshin16}. 
Consider $\rtY_i\sim\EC_{\bbm}(0,\sigma^2\Sigma_1,\Sigma_2,\dots,\Sigma_p,\varphi)$ for $i=1,2,\dots,n$. As described in Section \ref{sec:secproperties}, if we let $\rtZ_i = \rtY_i / ||\rtY_i||$ for all $i=1,2,\dots,n$, then $\vecc(\rtZ_i)\sim\AC_{\bbm}(\Sigma)$, and does not depend on the underlying type of EC distribution. 
Using this result along with the PDF in \eqref{eq:angular_pdf} means
that, regardless of $\varphi$, we can write the loglikelihood for $\Sigma_1,\Sigma_2,\dots,\Sigma_p$ as 
\begin{equation}\label{eq:tyfull}
\ell(\Sigma_1,\Sigma_2,\dots,\Sigma_p) = 
-\dfrac{n}{2} \log\left|\bigotimes_{k=p}^1\Sigma_k\right| 
-\dfrac{m}{2} \sum_{i=1}^n \log D^2_{\Sigma}(\mZ_i,0).
\end{equation}
We now optimize~\eqref{eq:tyfull} using a block-relaxation algorithm \citep{leeuw94} that partitions the parameter space into multiple blocks, and sequentially optimizes each block while fixing the remaining blocks. In this context, each block corresponds to $\Sigma_k$ under the constraint $\Sigma_k(1,1)=1$ for all $k=1,2,\dots,p$, and is derived in the following theorem :
\begin{theorem}\label{thm:tyrobustML}
For a fixed value of $\Sigma_{-k}$ and for $S_{ik} =
\mY_{i(k)}\Sigma_{-k}^{-1}\mY_{i(k)}^\top$, a fixed-point algorithm
for the maximum likelihood estimation of $\Sigma_k$ updates $\Sigma_{k,(t)}$, under the
constraint that $\Sigma_k(1,1) = 1$,  to
\begin{equation}\label{eq:tytensfixed}
\Sigma_{k,(t+1)} = ADJUST\left(n/m_{k},1,\sum_{i=1}^n \dfrac{S_{ik}}{\tr(\Sigma_{k,(t)}^{-1}S_{ik})}\right).
\end{equation}
\end{theorem}
\begin{proof}
See the Supplementary Material Section \ref{proof:tyrobustML}.
\end{proof}
\subsubsection*{ Estimation algorithm}
For each $k$, the fixed-point iteration algorithm of Theorem
\ref{thm:tyrobustML} depends on $\Sigma_{-k}$ only through
$S_{ik}$. Obtaining $S_{ik}$ for all $i=1,2,\dots,n$ and for any
$k=1,2,\dots,p$ has $O(nmm_M)$ computational complexity  (where $m_M=\max\{m_1,m_2,\dots,m_p\}$) while one iteration of \eqref{eq:tytensfixed} is usually considerably cheaper with a complexity of at most $O(m_M^3)$.  Therefore, in most cases we suggest performing $n_L$ iterations of \eqref{eq:tytensfixed} before going from $\Sigma_k$ to $\Sigma_{k+1}$. While all values of $n_L$ will eventually converge, we set $n_L=12$ based on numerical experiments. The final estimation algorithm is  as follows:
\begin{enumerate}
\item \textbf{Initialization:} Fit the TVN model of Section \ref{sec:tensnormalMLE} with lax convergence.
\item For $k=1,2,\dots,p$ 
\begin{enumerate}
\item Obtain $S_{ik}$ from Theorem \ref{thm:tyrobustML} $\forall i=1,\dots,n$
\item Update $\Sigma_{k,(t+1)}$  as per \eqref{eq:tytensfixed} for $n_L$ iterations.
\end{enumerate}
\item
\textbf{Iterate or stop:} Stop if convergence is achieved, or return to Step 2.
\end{enumerate}
\subsubsection*{Robust estimation of the location parameter}\label{subsec:robustM}
When $\mM$ is also unknown, \citet{tyler87a} and
\citet{maronna76} proposed (for the vector-variate case) the  fixed-point iteration algorithm:
$$
\mhM_{(t+1)} = (\sum_{i=1}^n \mX_i/D_{\Sigma}(\mX_i,\mhM_{(t)}))/(\sum_{i=1}^n 1/D_{\Sigma}(\mX_i,\mhM_{(t)}).
$$
\citet{tyler87a} was unable to show the joint existence of $\Sigma$
and $\mM$ due to discontinuities in the objective function, and so we
do not integrate this result into our estimation algorithm. However,
it can be easily integrated to estimate $\mM$ at each iteration and
then used to center the data before estimating
$\Sigma_1,\Sigma_2,\dots,\Sigma_p$. Although to our knowledge, no
theoretical guarantees exist for  this approach even in the
vector-variate case, we have found this algorithm to converge and produce stable estimates  (see Figures \ref{fig:ashoka} and \ref{fig:ashoka_diag}).

\section{Performance Evaluations}
\label{sec:performance}
In this section, we evaluate performance of the maximum likelihood estimation
procedures proposed in Section \ref{sec:MLE} by fitting them to
realizations  from the TVN and from gamma
scale TVN mixture (GSM) distributions. We used the GSM distribution  of
 Section \ref{sec:ECtensMLE}, with PDF as in
\eqref{dengengamma}, $(a,b)=(3,15)$, $i=1,2,\dots,n=10$ and
$(m_1,m_2,m_3,m_4,m_5,m_6) = (7,9,3,23,7,3)$.

  \begin{figure*}[h]
     \centering
     \includegraphics[width = \textwidth]{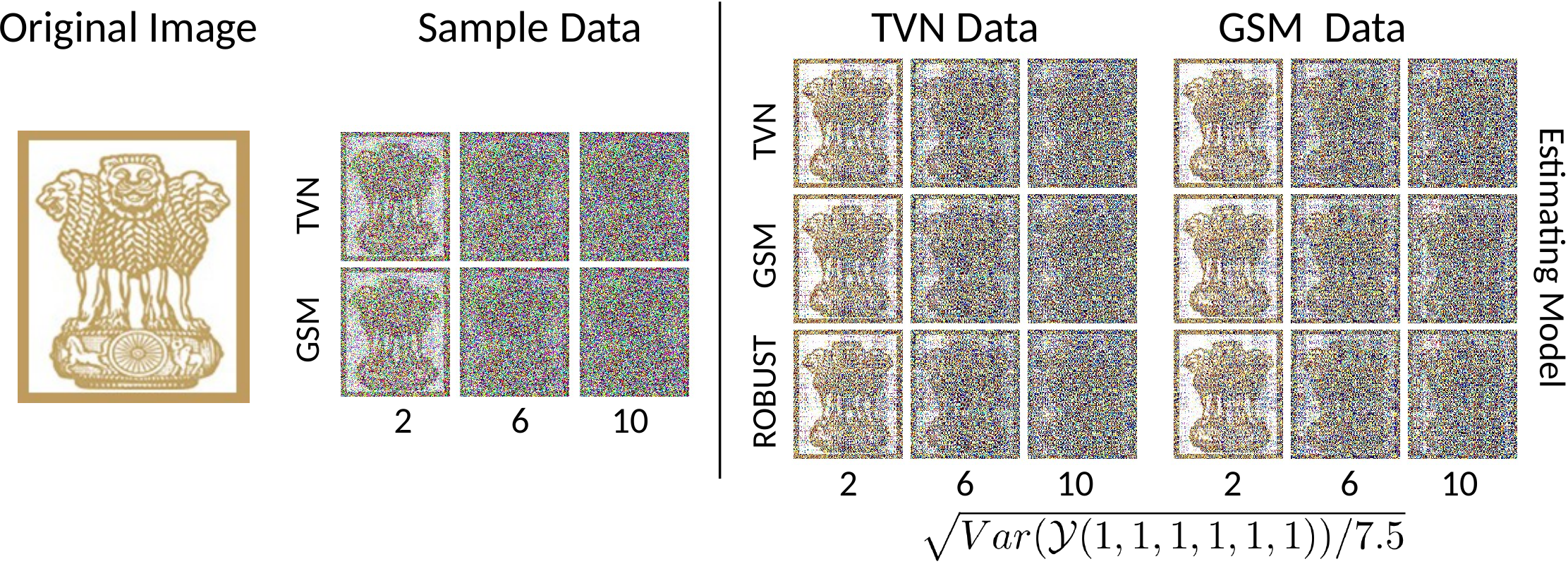}
     \caption{ \underline{\em Left block}: The Lion
       capital of Ashoka that is the ground truth of the rearranged mean
       tensor $\mM$ of size $7\times 9\times 3\times 23\times 7\times
       3$, and sample realization with increasing variance across
       two EC distributions. \underline{\rm Right
         block:}  Estimated $\mM$ asuming TVN, GSM models
       and Tyler's robust estimators, under TVN and GSM-simulated
       data.  The GSM model-assumed  and Tyler estimators 
       perform similarly as the TVN model-assumed estimators for
       TVN-simulated data, but better than the latter for        GSM-simulated data.} 
     \label{fig:ashoka}
\end{figure*}
We simulated $\Sigma_k$ (for $k=1,2,\dots,6$) from $\mW_{m_k}(100m_k,I_{m_k})$ or a  Wishart distribution, before constraining
its condition numbers to be at most 50 and scaling each entry by its $(1,1)$th element.  
  The tensor-valued $\mM$ was chosen such that when rearranged in a 3-way
  tensor of size $(7\times9\times 3)\times(23\times 7)\times 3 $, it corresponds to the
  three RGB slices of a $189\times 161$ true color image  of the Lion
  Capital of Ashoka, shown in the leftmost image of Figure \ref{fig:ashoka}.

We also simulated data for
 $\sigma\in\{2,6,10\}$ and for the TVN case that is a
special case of \eqref{eq:scaletensnorm} when $\Prob(Z_i=1)=1$ for all
$i=1,\dots,n$. Since $\Var(\vecc(\rtY_i)) = 15\sigma^2\Sigma$, 
 we set $\sigma^2$ for TVN-simulated data to be 15 times that of the
 GSM-simulated cases for both sets to have comparable overall
 variability. Figure~\ref{fig:ashoka} (second set of images in the
 left block) displays simulated realizations at each setting: as
 expected, increasing  $\sigma^2$ produces noisier  images.  
  Figure \ref{fig:ashoka} (right block) displays TVN-,
  GSM- and robust-estimated $\mhM$s. For Tyler's robust estimator we used the location  estimation procedure  of Section \ref{subsec:robustM}. 
For all cases, estimation
 performance worsens with increasing $\sigma$. For TVN-simulated data,
 there is little difference in the quality of $\mhM$ obtained by the
 three methods.  
But for heavy-tailed GSM-simulated data, the GSM or robust $\mhM$s
recover $\mM$ better than the TVN-estimated $\mhM$.  

Our illustration in Figure~\ref{fig:ashoka} was only on one
realization for each setting, so we account for simulation variability
in a larger study under a similar framework.
We  generated TVN and GSM data for $\sigma = 2,6,10$ and 100 different
realizations of $(\Sigma_1,\Sigma_2,\dots,\Sigma_6)$, each drawn as before from a Wishart distribution.\begin{figure*}[h]
     \centering
	 \includegraphics[width = \textwidth]{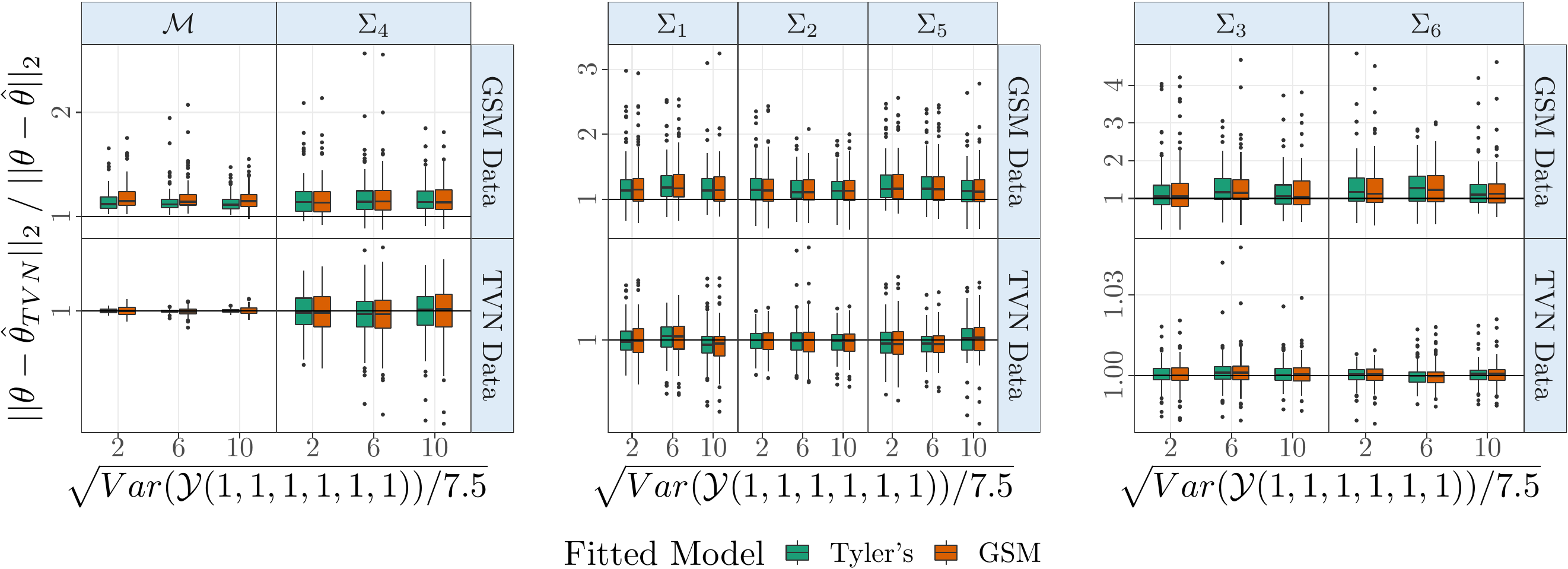}
     \caption{ Relative Frobenius norms of differences in parameters
       (generically denoted by $\theta$) for 
       $\sigma\in\{2,6,10\}$, TVN- and GSM-simulated data, and
       estimates obtained under Tyler's robust and GSM model
       assumptions. 
       Each plot is obtained over 100 realizations from
       the  GSM and TVN models with different values of 
     $\Sigma$s, themselves independent
     Wishart realizations. 
       All points are concentrated around 1 for TVN-simulated data, indicating  that the GSM and Tyler's estimates perform at least as well as the TVN model when the data is TVN-simulated. For the higher dimensional parameters $\mM$ and $\Sigma_4\in\mathbb{R}^{23\times 23}$ and for GSM-simulated data, nearly all points are above 1, indicating that the GSM and Tyler's models outperform the TVN model in this scenario.  For medium-sized parameters $\Sigma_1,\Sigma_2\in \mathbb{R}^{7\times 7}$ and $\Sigma_5\in\mathbb{R}^{9\times 9}$ and for GSM-simulated data, most points are also larger than 1. However, for the smaller $3\times 3$ matrices $\Sigma_3$ and $\Sigma_6$ and for GSM-simulated data, the points are more concentrated at around 1. In conclusion, the GSM and Tyler's models perform at least as well as the TVN model for TVN-simulated data. Morover, the GSM and Tyler's models outperform the TVN model for GSM-simulated data, and their relative performance increase with larger covariance matrices.
    }
     \label{fig:ashoka_diag}
   \end{figure*} 
Figure \ref{fig:ashoka_diag} displays the relative Frobenius norms of the
difference between the true $(\mM,\Sigma_1,\Sigma_2,\dots,\Sigma_6)$ and the estimated $(\hat\mM,\hat\Sigma_1,\hat\Sigma_2,\dots,\hat\Sigma_6)$ parameters, 
 obtained under the GSM, TVN and Tyler models, and for 100 independent samples realized under  GSM and TVN assumptions. 
Specifically, the relative Frobenius normed difference  is $R_d =||\theta - \hat\theta_{TVN}||_2
/||\theta - \hat\theta||_2$, where $\theta$ generically denotes the
parameter being compared, $\hat\theta_{TVN}$ corresponds to the
estimated parameter under the TVN model and $\hat\theta$ denotes the
estimate of the parameter obtained under the assumed GSM or Tyler's
model. Therefore, values larger than unity indicate better performance
of the GSM (or Tyler) estimator over that obtained under  TVN
assumptions. Conversely, values lower than unity indicate worse
performance of GSM (or Tyler) estimates than those obtained under TVN
assumptions.  Our display in Figure~\ref{fig:ashoka_diag} is grouped by parameters of similar
dimensionality, with the first plot displaying results for $\mM$ and $\Sigma_4\in \mathbb{R}^{23\times 23}$, the second one for $\Sigma_1,\Sigma_5 \in \mathbb{R}^{7\times 7}$ and $\Sigma_2\in \mathbb{R}^{9\times 9}$, and the last one for displaying values of the $3\times 3$ matrices $\Sigma_3$ and $\Sigma_6$.
For TVN-simulated data,  there is not much to choose from betwen the
TVN-assumed estimates and those obtained using  GSM or Tyler
estimators. 
 On the other hand, for GSM-simulated data and the higher dimensional
 parameters $\mM$ and $\Sigma_4$, the  robust GSM and Tyler models
 outperform the TVN model, since nearly all values exceed unity. The
 performance of the TVN model is still worse for
 $\Sigma_1,\Sigma_2,\Sigma_5$ when compared to the GSM and Tyler
 models, since $R_d$ is then mostly above 1. However, for smaller
 $3\times 3$ covariance matrices ($\Sigma_3$ and $\Sigma_6$), the TVN
 model performs more like the GSM and Tyler models. In conclusion, the
 GSM and Tyler fits perform similarly to the TVN model for
 TVN-simulated data. Morover, the GSM and Tyler models outperform the TVN model when the data are GSM-simulated, and their performance relative to the TVN increase with larger covariance matrices.
Our simulation experiment shows the importance of considering an EC
distribution that is more general than the TVN when the tails of the
data are heavier. It also shows the benefit of using Tyler's robust
estimator even when the exact EC distribution is unknown.

\section{Applications to Image Learning}\label{sec:application}
We apply our EC tensor-variate estimation methodology to two
learning applications involving images. The first application is
in the context of improved prediction of an image class using tensor-variate discriminant analysis and classification while the
second application characterizes the distinctiveness of face images in
terms of gender, age and ethnic 
origin. In each case, the EC distribution with heavier
tails (here, the TV-$t$ distribution) is demonstrated to have better
results than its TVN counterpart.
\subsection{Discriminant analysis for image classification}\label{sec:classification}
In this section, we provide a general linear (LDA) and quadratic (QDA)
framework for the classification of tensor-valued data, in a similar
manner as done for vector-~\citep{andersonandbahadur62} and
matrix-variate \citep{thompsonetal20} data. We also illustrate how our
maximum likelihood estimation methods of Section \ref{sec:MLE} are incorporated into
the LDA and QDA frameworks, and the benefit of using a more general EC
tensor-variate distribution than the TVN in predicting images of dogs
or cats in a two-class digital image classification problem.
The development of linear and quadratic discriminant analysis is straightforward, so we refer to Section~\ref{sec:classmethod}
for details. The methods in Section~\ref{sec:MLE} are used to estimate the parameters in each class population. Our framework can extend other discriminant analysis methods that considered the normal distribution \citep{panetal19}.

\begin{figure*}[h]
\subfloat[\label{subfig:catndogimag}]{
\includegraphics[width=\linewidth]{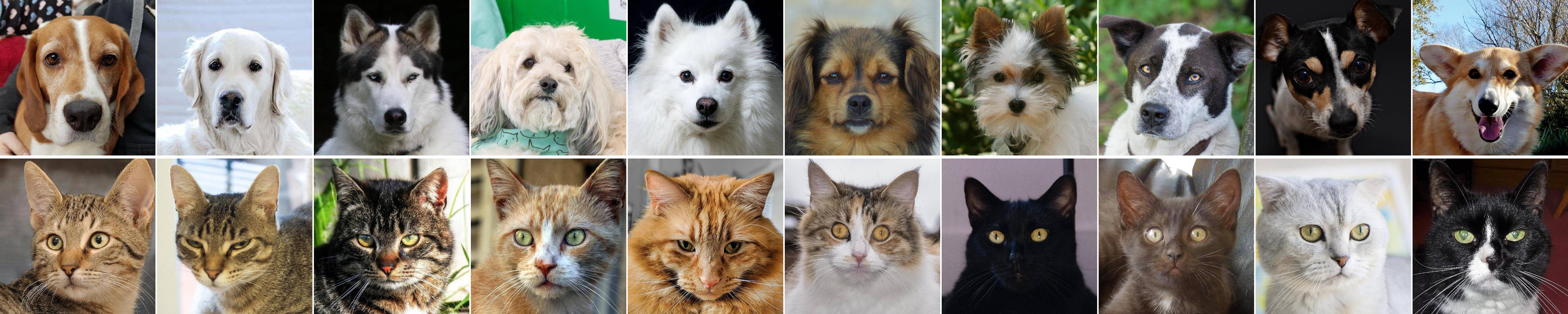}
}
\\
\subfloat[\label{subfig:curves}]{
\includegraphics[width=\linewidth]{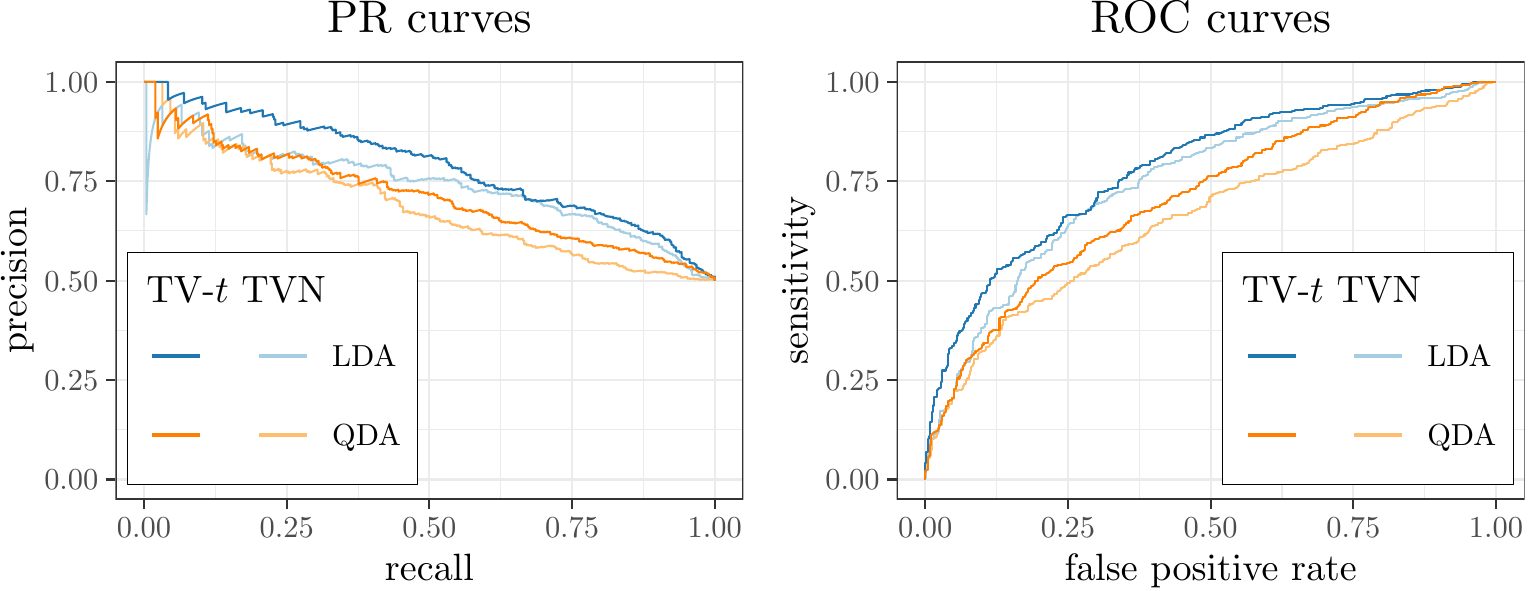}
}
\caption{
(a) sample dog and cat images from the AFHQ dataset,
(b) precision-recall (PR) and  receiver operating characteristic (ROC) curves for the tensor-variate-$t$ (TV-$t$) model and the TVN model, and  across both LDA and QDA classification frameworks. The TV-$t$ model outperforms the TVN model in terms of ROC and PR, since the curves corresponding to the TV-$t$ model are higher than those that correspond to the TVN model.
}
  \centering
\end{figure*}
The Animal Faces-HQ (AFHQ)
dataset \citep{choietal20} has 15,000 512$\times$512 digital photographs of wild and
domestic felines and canines. In this paper, we use our developed
methodology to perform classification on the 5,653 cat and 5,239 dog
images from the AFHQ dataset, of which 500 images of each animal are
test datasets \citep{choietal20}. 
Figure \ref{subfig:catndogimag} shows images of ten randomly selected
dogs and cats from the database. Despite the 
attempts at alignment, the images have dissimilar lighting conditions
and viewing angles, so we sequentially applied the Radon~\citep{radon86}
and 2D discrete wavelet transforms (DWT)~\citep{daubechies92} to the data in each of the
RGB channels and extracted the low-frequency (LL) components~\citep{jafariandsoltanian05,jadhavandholambe09}, yielding a $22\times 22\times 3$ array of features for each
training and test image. Figure 
\ref{fig:radondwt_animals} displays the sequentially applied Radon and
DWT transforms that correspond to the cat and dog images of Figure
\ref{subfig:catndogimag}. The processing extracts the local and
directional features in the images, with the Radon transform improving
 the low-frequency components~\citep{jadhavandholambe09} that are then
 extracted by the  wavelet transform~\citep{jafariandsoltanian05}.
We use QDA and LDA to build discriminant rules from these transformed
array data, by modeling each population of transformed images via the
TV-$t$ distribution  with unknown DF,  and for
comparison, also the TVN distribution. In particular, we  fit 
\begin{equation}\label{eq:totr_catndogs}
\mY_i = \langle \bx_i | \mB\rangle + \mE_i
,\quad 
\mE_i\sim t_{22,22,3}(\nu;0,\sigma^2\Sigma_1,\Sigma_2,\Sigma_3),
\end{equation}
where $\mY_i\in\mathbb{R}^{22\times 22\times 3}$ is the $i$th
transformed  image and  $\bx_i$ is $[1,0]'$ or $[0,1]'$ depending on
whether $\mY_i$ is a transformed image of a cat or a dog. The tensor-variate
errors $\mE_i$ are set to follow a TV-$t$ distribution with unknown
$\nu$ DF, which corresponds to the GSM distribution of
\eqref{dengengamma} for $a=b=\nu$. Based on our model
\eqref{eq:totr_catndogs}, the densities $f_{cat}$ and $f_{dog}$  involved in the classification rule of \eqref{eq:eqda_general} correspond to the EC distributions $t_{22,22,3}(\nu_{cat};\mB_{cat},\sigma^2\Sigma_1,\Sigma_2,\Sigma_3)$ and $t_{22,22,3}(\nu_{dog};\mB_{dog},\sigma^2\Sigma_1,\Sigma_2,\Sigma_3)$, where $\mB_{cat} = \mB\times_1 [1,0]$ and $ \mB_{dog} = \mB\times_1 [0,1]$ are the mean parameters of the transformed cat and dog image populations. 
Similarly for QDA, we estimate $f_{cat}$ and $f_{dog}$ by fitting 
$t_{22,22,3}(\nu_{cat};\mB_{cat},\sigma^2\Sigma_1,\newline\Sigma_2,\Sigma_3)$ and
$t_{22,22,3}(\nu_{dog};\mB_{dog},\sigma^2\Sigma_1,\Sigma_2,\Sigma_3)$ models
using the ToTR methodology of \eqref{eq:totr_catndogs} with
$\bx_i\equiv 1$ always.

These models (one for LDA and two for QDA) were fit to the training
set images of $5153$ cats and $4739$ dog images using the ECME
algorithm of Section \ref{sec:totr_ec}. The DFs $\nu$ were estimated
using the either step of \eqref{eq:eithernu} to be $\nu = 3.33$ for
the LDA case, and $\nu_{cat}=3.05$ and $\nu_{dog}=3.51.$ We 
specified $\Sigma_1$ and $\Sigma_2$ to be AR(1) correlation matrices
to capture  spatial context in the transformed images, and $\Sigma_3$
to have an equicorrelation structure to denote correlation between the
three RGB channels. This autocorrelation was estimated to be 0.94 for
the LDA case and the same for the cat transformed images but
marginally lower (0.93) for the dogs in the QDA fits. We also imposed low-rank CP formats on
$\mB$, with ranks chosen from among five candidates using cross
validation (CV) on the training data, and both missclassification rate
and area under the curves (AUC) as our CV decision criteria. For
comparison, we also performed QDA and LDA under the TVN model. 

The classification rule of \eqref{eq:eqda_general} for a transformed
image $\mY$ is specified in terms of the loglikelihood ratio (LLR) $=
\log(f_{cat}(\mY))-\log(f_{dog}(\mY)) + log(\eta_1/\eta_2)$. A large
LLR value indicates high posterior probability that $\mY$ is a cat image,
with small LLR values conversely indicating high posterior probability
that $\mY$ is the image of a dog. A zero value indicates complete
uncertainty in the classification. In our application, we set
$\eta_1=\eta_2=0.5$. 
 We obtained posterior logit-probabilities for the $500$ cat and $500$ dog images in the test set, and used them to calculate the precision-recall (PR) curves \citep{davisandgoadrich06} and the receiver operating characteristic (ROC) curves \citep{bradley97,bantisandfeng16} of Figure \ref{subfig:curves}. 
 We observe that for both LDA and QDA, PR curves are higher with the
 TV-$t$ than with the TVN models. This means that there is higher
 precision at all possible recall values using the TV-$t$ model than
 the TVN model. In the QDA PR curves, the TVN model sometimes has
 higher precision than the TV-$t$ 
 model (for recall values of less than 0.2). However, thresholds with
 recall values of less than 0.2 are not  practical, since they lead 
 false negatives (that is, cat images being classified as dog images) to be four
 times more  numerous than true positives (correctly classified cat images).  
 Similarly, in all cases the ROC curve is higher for the TV-$t$ model
 than for the TVN model. Higher ROC curves mean that all possible
 false positive rates have  higher sensitivity (true positive rate)
 for the TV-$t$ than for the TVN case.  We finish this section noting
 that $\nu$  was estimated to be less than 3.6 in all cases, meaning that
  our tensor-variate data are quite heavy-tailed. This explains why the
  TV-$t$ outperforms the TVN distribution in terms of classification
  for this real-data example, since it is able to accommodate the
  heavier tails of the transformed image data.

\subsection{Distinguishing facial characteristics}\label{lfw}
The LFW database is commonly used in the development and testing of facial
recognition methods \citep{huangetal07}, and consists of over 13,000
$250\times 250$ color images. Each image 
has the labeled attributes of ethnic origin, age group and
gender~\citep{afifiandAbdelhamed19,kumaretal09}. We use our ToTR
methodology of Section \ref{sec:totr_ec} to distinguish the
visual characteristics of different attributes. Specifically, we perform a 3-way tensor-variate analysis of variance (TANOVA) model to
distinguish faces across the three factors of gender (male or 
female), continent of ethnic origin (African, European or Asian, as
specified in the database), and cohort (child, youth, middle-aged or
senior).  Of the more than 13,000 images in the LFW database, we use
the 605 images with unambiguous genders, age group and ethnic origin,
and with at most 33 images for each factor combination that were
selected and made available by  \citet{llosaandmaitra20}. The authors also
cropped the images to a central region of size $151\times 111$ each,
and logit-transformed them to match the statespace of the normal
distribution. In this application, we evaluate whether the use of a
model with TV-$t$ errors can provide a more improved fit of the model.

We fit the 3-way TANOVA model with TV-$t$ errors:
\begin{equation}\label{eq:lfwtotr}
\mY_{ijkl} = \langle \mX_{ijk}|\mB\rangle + \mE_{ijkl}
,\quad 
\mE_{ijkl} \sim t_{\bbm_3}(\nu;0,\sigma^2\Sigma_1,\Sigma_2,\Sigma_3),
\end{equation}
where $\bbm_3 = [151,111,3]'$ and $\mY_{ijkl}\in \mathbb{R}^{151\times
  111\times 3}$ is the 151$\times$111 RGB color image that corresponds
to the $l$th  person of the $i$th gender, $j$th ethnic origin and
$k$th age group  $(i=1,2;j=1,2,3;k=1,2,3,4;l=1,\sdots,n_{ijk})$.  Here $\mX_{ijk}\in\mathbb{R}^{2\times3\times
  4}$ is the tensor-valued covariate with 1 at position $(i,j,k)$ and
zero elsewhere, and indicates the factor combination of
$\mY_{ijk}$. The regression coefficient
$\mB\in\mathbb{R}^{2\times3\times4\times151\times111\times3}$ contains
the group means across all factor combinations. As in
Section~\ref{sec:classification}, we constrain  $\Sigma_1$ and
$\Sigma_2$ to be AR(1) correlation matrices and $\Sigma_3$ to be an
equicorrelation matrix. We fit our model \eqref{eq:lfwtotr} assuming
TV-$t$ errors with unknown DF parameter $\nu$ and TVN errors. The DF $\nu$ was estimated in all cases, and for ultra heavy-tailed data $\nu$ was constrained to never be lower than 2.01, so that that the first two population moments exist. Such is the case for this LFW dataset, which selected $\hat\nu=2.01$ as the DF in all cases.
We also fit the model assuming two different assumptions on the
$\mB$: the CP format, the TT format~\citep{oseledets11} that is the
tensor-ring (TR) format \citep{zhaoetal16a} when one of the TR ranks
is set to 1. We also assumed an  unconstrained $\mB$ (with no
reduction in the number of 
parameters). We chose the CP and TT ranks using the Bayesian
Information Criterion~(BIC)~\citep{kashyap77,schwarz78} out of a total
of 26 candidate ranks (for CP), and from nine ranks 
for the TT format. The chosen CP rank was $R=70$, with  19,175
unconstrained parameters in $\mB$, or a 98.4\% 
dimension reduction relative to the unconstrained $\mB$ that contains more than 1.2 million parameters. The
TT rank was chosen to be $g=(1,2,4,8,4,2)$, which resulted in a total
of 8,450 unconstrained parameters in $\mB$, or a 99\%
dimension reduction relative to the unconstrained $\mB$. Table \ref{table:BIC_lfw} displays the BICs
across the six (two error distributions across three $\mB$ formats)
assumptions.
\begin{table}[h]
\centering
  \caption{BICs for different fits, with smaller values indicating
    better model fit. 
    The TV-$t$ fits in this table and this paper have a selected DF parameter $\nu$ of 2.01, which in all cases outperformed their TVN counterparts.}
\begin{tabular}{ c  c | c  c c }
 &  & \multicolumn{3}{c}{Format of $\mB$}\\
& & TT  & CP & none  \\     \hline
 \multirow{2}{*}{Error Model}
 &TVN& -11,190,861 & -11,223,595 & -4,745,807\\ 
&TV-$t$ & -23,992,185& -24,064,833 & -20,043,342
     \end{tabular}
\label{table:BIC_lfw}
\end{table}

From Table \ref{table:BIC_lfw}, we observe that
the TV-$t$ distribution always outperforms its TVN counterpart, which
indicates the benefit of considering a model for the heavy
tails. Moreover, within the TVN or TV-$t$ models, the CP
and TT formats always outperform the model with unconstrained $\mB$,
with the CP always performing better than the TT. This shows the
benefit of using a low-rank format that allows massive parameter
reductions while also preserving model accuracy. Importantly in 
the context of this paper, we find that the TV-$t$ distributions with
constrained or unconstrained $\mB$ outperform all the TVN models,
including the TT and CP formats. Our findings  show that performing
parameter reduction 
through the CP or TT low-rank formats should be accompanied with the
more appropriate model for the tail weights. Indeed, the best performance is
achieved when the low-rank format is used in combination with the
heavy-tailed TV-$t$ distribution.

\begin{figure*}[h]
\subfloat[TV-$t$ \label{subfig:LFWt}]{
\includegraphics[width=0.49\linewidth,page = 2]{figures/faces-crop.pdf}
}
\subfloat[TVN \label{subfig:LFWn}]{
\includegraphics[width=0.49\linewidth,page = 1]{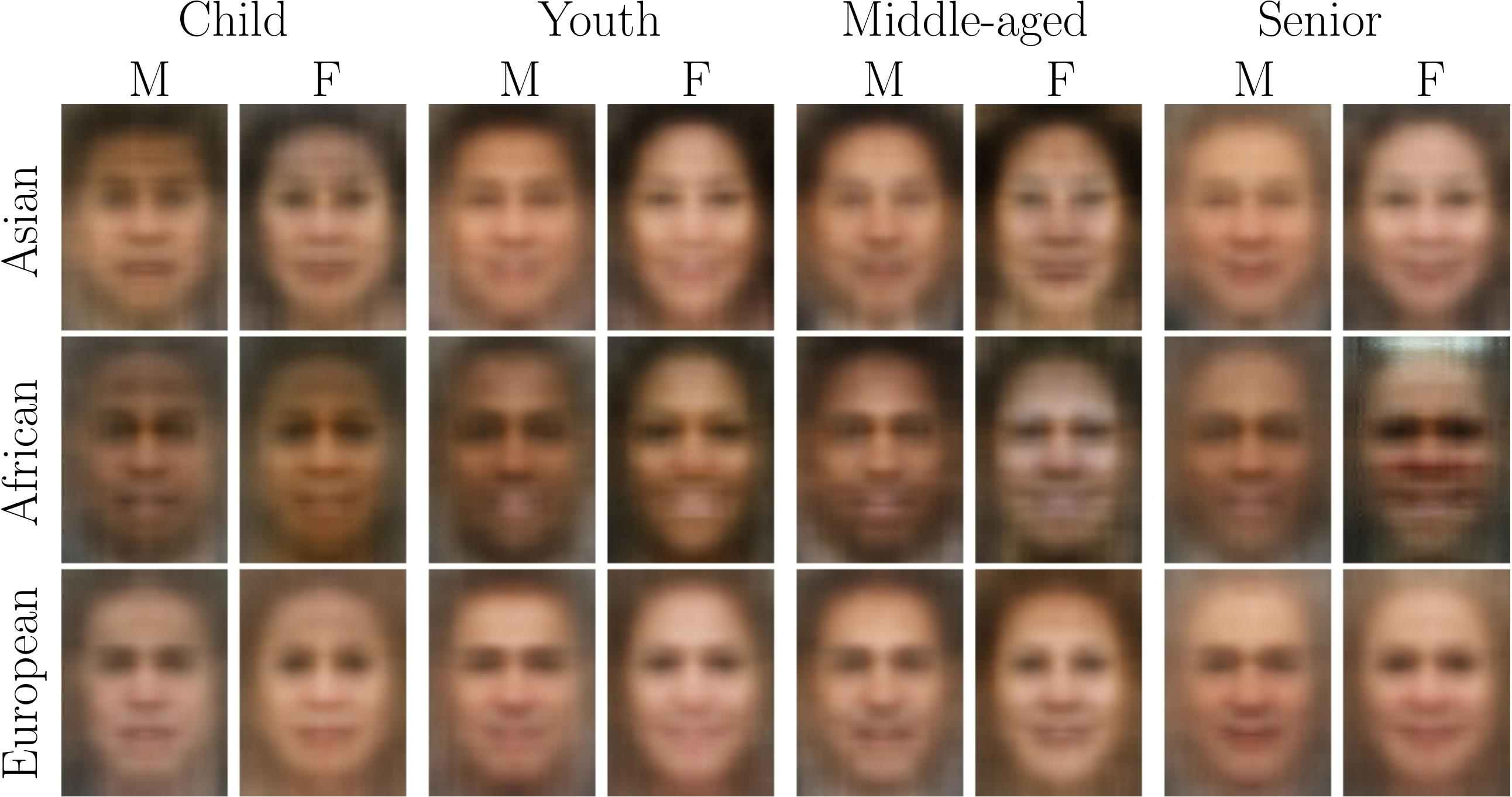}
}
\caption{
Different slices of the resulting factorized tensor $\mB$ that results from fitting the ToTR model on a 3-way factor ANOVA of the LFW dataset using the CP format and the (a) TV-$t$ distribution and (b) TVN distribution. The results are CP-format compressed mean images across age groups (child, youth, middle aged and senior), ethnic origin (Asian, African, European) and  genders (male, female) from the 605 LFW images used in \citep{llosaandmaitra20}. The imposed CP format preserved vital information regarding the factor-combination of age group, gender and ethnic-origin while providing substantial (more than $98\%$) parameter reduction.
}
  \centering
\end{figure*}
Figures \ref{subfig:LFWt} and \ref{subfig:LFWn}  display the estimated
$\mhB$ for the TV-$t$ ($\nu=2.01$) and TVN models. In both
cases, the CP format preserves vital visual characteristics regarding ethnic
origin, gender, and age-group. However, the TVN formatted $\mB$ results in
more exaggerated facial expressions, while the fits under
the TV-$t$ model are sometimes noisier, but more accurately reflect
the diversity of facial characteristics in each sub-group. Therefore,
we contend that the TV-$t$ model not only provides a quantitatively
better fit to the TANOVA model, but also more accurately displays the
heterogeneity in the dataset.

\section{Discussion}
\label{sec:discussion}
In this article, we introduced, 
defined, characterized and studied in detail the family of spherically
and EC tensor-variate distributions that generalize the well-studied TVN distributions. We derived CFs, PDFs, exact distributions of linear
reshapings, linear transformations, 
derived the moments of different orders as well as of special forms,
and introduced the EC tensor-variate Wishart distribution. These
properties arise not only from the vectorized elliptical form of the
random tensor, but also from its intrinsic tensor-variate structure, which in many cases makes  estimation and inference possible.
We provided algorithms for maximum likelihood estimation
under different assumptions on the observed EC data. First, we suggested a modification to the  maximum likelihood estimation  algorithm for TVN data that makes its parameter estimates identifiable. We then use this algorithm to propose a maximum likelihood estimation algorithm for uncorrelated EC tensor-valued observations. We also derive EM
algorithm variants for maximum likelihood estimation under IID draws from a
scale mixture of TVN distribution, and a novel, robust and constrained Tyler-type
algorithm for maximum likelihood estimation when the underlying EC tensor-variate distribution is
unknown. We further proposed a ToTR framework with EC errors that extends the work of \citet{llosaandmaitra20} to the case where the tensor-response regression has  EC tensor-variate distributed errors under various low-rank format assumptions on the regression coefficient.
We studied the performance of our estimation algorithms through a
simulation experiment on 6-way tensor-valued data that showed the
limitations of the TVN distribution when the data have heavier tails
than can be modeled by the TVN. Indeed, we demonstrated that fitting a
GSM or Tyler robust estimating model to such data results in improved performance when the data is not Gaussian.
 Further, we provided methodology that allows us to use our maximum likelihood methodology
 to perform LDA and QDA on tensor-data classification and
 discrimination, and applied it to predict images of cats or dogs in the AFHQ dataset. We demonstrate that the AFHQ data is quite heavy-tailed, resulting in better performance for the TV-$t$ model relative to the TVN model in terms of ROC and PR curves. 
  Finally, we used our ToTR methodology to perform
 a 3-way TANOVA of facial images from the LFW database.  
 We demonstrate that the database is exceptionally heavy-tailed (with a selected DF of 2.01), since the
 TV-$t$ models outperfom their TVN counterparts in all scenarios.

There are several possible extensions and generalizations of our
work. First, we can define spherical distributions in multiple ways. The class of distributions $
\mathscr{F}_k = \{ \rtX \!:\!\rtX\overset{d}{=} [\![\rtX;\Gamma_1,\sdots\Gamma_k,I_{m_{k+1}},\newline \sdots,I_{m_p}]\!]$,  $\Gamma_j\in O(m_j) \forall j\in\{1,\sdots,k\}\}$
is such that 
$
\mathscr{F}_A \subseteq
\mathscr{F}_{p} \subseteq
\mathscr{F}_{p-1} \subseteq
\sdots \subseteq
\mathscr{F}_{1},
$
where $\mathscr{F}_A$ is the family we defined in this article.
All these families have Kronecker-separable
covariance structures if the family of EC random tensors is defined
using the Tucker product, as in \eqref{eq:defscale}. The
Kronecker-separable structure makes it practical for us to
explore the parameter space of the ultra-high dimensional covariance
matrix. However, other types of variance structures can also be
considered. Finally, the sampling distributions of  $\mB$ in our ToTR with EC errors are unknown, and may be
of interest in applications where significance testing is
important. Thus, we see that though we have developed theory and estimation
methodology for EC tensor-variate distributions, there remain a
number of extensions and generalizations that are worthy of
investigation.

\section*{Acknowledgments}
A version of this paper won C. Llosa-Vite an award at the 2022
Statistical Methods in Imaging Student Paper Competition. 
This research was supported in part by the National Institute of
Justice (NIJ) under Grants No. 2015-DN-BX-K056 and 2018-R2-CX-0034. The research of the  second author was also supported in part by the National
Institute of Biomedical Imaging and Bioengineering (NIBIB) 
of the National Institutes of Health (NIH) under Grant R21EB016212,
and the United States Department of Agriculture (USDA) National
Institute of Food and Agriculture (NIFA) Hatch project IOW03617.
The content of this paper is however solely the responsibility of the
authors and does not represent the official views of the NIJ, the NIBIB, the NIH, the NIFA or the USDA.

\section*{\Large Supplementary Appendix}

\renewcommand\thefigure{S\arabic{figure}}\setcounter{figure}{0}
\renewcommand\thetable{S\arabic{table}}\setcounter{table}{0}
\renewcommand\thesection{S\arabic{section}}\setcounter{section}{0}
\renewcommand\theequation{S\arabic{equation}}\setcounter{equation}{0}

\section{Supplement to Section 2 }

\subsection{Tensor algebra notation involved in Section \ref{sec:prelim}}
\label{app:tensdef}
As in Section \ref{sec:prelim}, we assume $\mX$ is a $p-$way tensor of size $m_1\times m_2\times \hdots\times m_p$ and $\boldsymbol{e}_{i}^{m}\in \mathbb{R}^{m}$ is a unit-basis vector with 1 as the $i$th element and 0 elsewhere. The vector outer product of a set of vectors $(\bz_1,\bz_2,\hdots,\bz_p)$, where $\bz_j\in \mbR^{n_j}$ for all $j=1,2,\hdots,p$, is written as $\circs_{j=1}^p \boldsymbol{z}_j$ and results in a $p-$way tensor of size $n_1\times n_2\times\hdots\times n_p$ such that  $(\circs_{j=1}^p \boldsymbol{z}_j)(i_1,i_2,\hdots,i_p)= \prod_{j=1}^p \boldsymbol{z}_j(i_j).$ This product is useful for expressing $\mX$ as
\begin{equation}\label{elementform}
\sum_{i_1=1}^{m_1} \hdots \sum_{i_p=1}^{m_p} \mX(i_1,i_2,\hdots,i_p)
\big( \circs_{q=1}^p \boldsymbol{e}_{i_q}^{m_q} \big),
\end{equation}
and the Tucker product between $\mX$ and $\sfA_1,\sfA_2,$ $\hdots, \sfA_p$ as
\begin{equation*}\label{tuckerform}
[\![ \mX; \sfA_1,\sfA_2, \hdots,\sfA_p  ]\!]
=\sum_{i_1\dots i_p} \mX(i_1,i_2, \hdots, i_p)
\big( \circs\limits_{q=1}^p \sfA_q(:,i_q) \big),
\end{equation*}
where $\sfA_q(:,i_q)$ is the $i_q$th column of $\sfA_q$. 
Tensor reshapings are obtained by manipulating the vector outer product. We define the vectorization, $k$th mode matricization and $k$th canonical matricization of $\mX$ as
\begin{equation}\label{defvec}
\vecc (\mathbf{\mX})
= \sum\limits_{i_1=1}^{m_1} \hdots \sum\limits_{i_p=1}^{m_p} \mX(i_1,i_2, \hdots, i_p)
\big( \bigotimes\limits_{q =p}^1 \boldsymbol{e}_{i_q}^{m_q} \big),
\end{equation}
\vspace*{-.3cm}
\begin{equation}\label{kmode}
\mX_{(k)} 
=\sum\limits_{i_1=1}^{m_1} \hdots \sum\limits_{i_p=1}^{m_p} \mX(i_1,i_2, \hdots, i_p)
\boldsymbol{e}_{i_k}^{m_{k}} \big( \bigotimes\limits_{q=p,q\neq k}^1 \boldsymbol{e}_{i_q}^{m_q} \big)',
\end{equation}
\vspace*{-.3cm}
\begin{equation}\label{canonical}
\mX_{<k>}  
= \sum\limits_{i_1=1}^{m_1} \hdots \sum\limits_{i_p=1}^{m_p} \mX(i_1,i_2, \hdots, i_p)
\big( \bigotimes\limits_{q =k}^1 \boldsymbol{e}_{i_q}^{m_q} \big)
\big( \bigotimes\limits_{q =p}^{k+1} \boldsymbol{e}_{i_{q}}^{m_{q}} \big)', 
\end{equation}
respectively for any $k=1,2,\hdots,p$.  
The tensor outer product ($\circ$) between $\mX$ and $\mY \in \mathbb{R}^{n_1 \times n_2\times\hdots\times n_q}$ results in a tensor $\mX\circ\mY$ of size $(m_1 \times m_2\times\hdots\times m_p\times n_1 \times n_2\times\hdots\times n_q)$ such that
$
(\mX\circ\mY)(i_1,i_2,\hdots,i_p,j_1,j_2,\hdots,j_q) = \mX(i_1,i_2,\hdots,i_p)\mY(j_1,j_2,\hdots,j_q).
$
With the tensor outer product defined, we now state and prove
\begin{lemma}\label{lemma:outervectens}
Let $\rtY\in \mathbb{R}^{n_1\times n_2\times \dots\times n_q}$. Then 
$(\rtX\circ\rtY)_{<p>} = (\vecc \rtX)(\vecc \rtY)'.$
\end{lemma}

\begin{proof}
See Section \ref{proof:outervectens}.
\end{proof}

\subsection{Marginal distributions}\label{sec:marginals}
The marginal distributions follow from Theorem \ref{ellipt:distform}. by choosing $\sfA_1,\sfA_2,\hdots,\sfA_p$ appropriately. As an example, consider the 3-way random tensor $\rtY \sim \EC_{m_1,m_2,m_3}\big(\mM,\Sigma_1,\Sigma_2,\Sigma_3 ,\varphi\big)$, which will be split into eight sub-tensors as
\begin{equation}\label{eq:subtens}
\rtY_{ijk} = 
[\![\rtY; \sfA_{i1},\sfA_{j2},\sfA_{k3}]\!],
\quad
\mM_{ijk} = 
[\![\mM; \sfA_{i1},\sfA_{j2},\sfA_{k3}]\!],
\end{equation}
where $\sfA_{1l}=[\sfI_{n_l}\vdots \bzero]$ and
$\sfA_{2l}=[\bzero\vdots\sfI_{m_l-n_l}]$ are block matrices of size
$n_l\times m_{l}$ and $(m_l-n_l)\times m_{l}$ respectively, and $ n_l<
m_{l}$ for all $l=1,2,3$ and $i,j,k=1,2$ (here $\sfI_h$ denotes an
$h\times h$ identity matrix and $\bzero$ is a matrix of zeroes). Figure \ref{fig:subtensor} 
\begin{figure}[h!]
\includegraphics[width=0.3\textwidth]{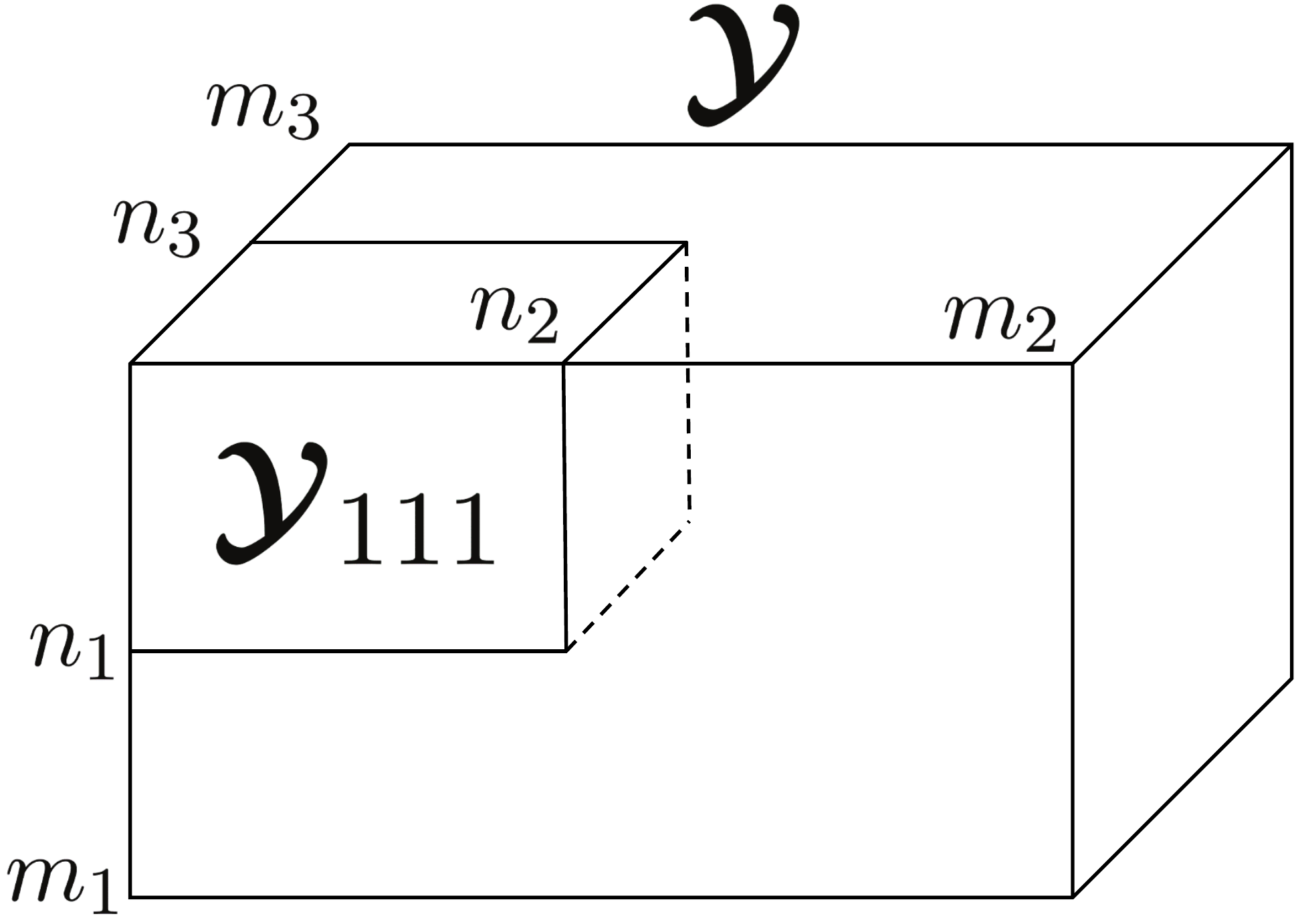}
  \centering
  \caption{A random third order tensor $\rtY$ along with one of its subtensors $\rtY_{111}$.}
  \label{fig:subtensor}
\end{figure}
displays $\rtY$ along one of its eight sub-tensors $\rtY_{111}$. Using Theorem \ref{ellipt:distform} and \eqref{eq:subtens}, we obtain its marginal distribution as
$$\rtY_{111}\sim
\EC_{n_1,n_2,n_3}\big(\mM_{111},\sfA_{11}\Sigma_1\sfA_{11}^\top,\sfA_{12}\Sigma_2\sfA_{12}^\top,\sfA_{13}\Sigma_3\sfA_{13}^\top,\varphi\big),
$$
where $\sfA_{1k}\Sigma_k\sfA_{1k}^\top$ is a sub-matrix corresponding to the first $n_k$ columns and rows of $\Sigma_k$.


\subsection{Conditional distributions along multiple modes}\label{supp:condsupp}
Here we demonstrate how to find conditional distribution along multiple modes.
As an example, consider the subtensor $\rtY_{111}$ of $\rtY \sim \EC_{m_1,m_2,m_3}\big(\mM,\Sigma_1,\Sigma_2,\Sigma_3 ,\varphi\big)$, both shown in Figure \ref{fig:subtensor} and let
\begin{equation*}
\Sigma_i = 
\begin{bmatrix}
\Sigma_{i,11}  & \Sigma_{i,12}\\
\Sigma_{i,21}  & \Sigma_{i,22}
\end{bmatrix}
,\quad
\Sigma_{i,11} \in \mathbb{R}^{n_i\times n_i},
\end{equation*}
and $\Sigma_{i,\{11\}|\bullet} = \Sigma_{i,11}-\Sigma_{i,12}\Sigma_{i,22}^{-1}\Sigma_{i,21}$ for $i=1,2,3$.
Then 
$$
\mL(  \rtY_{111} \big| \rtY_{\backslash\{111\}} = \mY_{\backslash\{111\}}
  ) =
\EC_{n_1,n_2,n_3}\big(\mM_{\{111\}|\textbackslash\{111\}},\Sigma_{1,\{11\}|\bullet},\Sigma_{2,\{11\}|\bullet},\Sigma_{3,\{11\}|\bullet},\varphi_{q} \big),
$$
where the event 
$\rtY_{\backslash\{111\}}=\mY_{\backslash\{111\}}$ means that  $
\rtY_{112}
= \mY_{112},\rtY_{121} = \mY_{121},\rtY_{122} = \mY_{122},\rtY_{211} =
\mY_{211},\rtY_{212} = \mY_{212},\rtY_{221} = \mY_{221},\mbox{ and }
\rtY_{222} = \mY_{222}$, 
\begin{equation*}
\begin{split}
\mM_{\{111\}|\textbackslash\{111\}}
&=
[\![(\mY_{222}-\mM_{222});\Sigma_{1,12}\Sigma_{1,22}^{-1},\Sigma_{2,12}\Sigma_{2,22}^{-1},\Sigma_{3,12}\Sigma_{3,22}^{-1}]\!]
\\+&
(\mY_{211}-\mM_{211})\times_1(\Sigma_{1,12}\Sigma_{1,22}^{-1})
-
[\![(\mY_{122}\!-\!\mM_{122});I_{n_1},\Sigma_{2,12}\Sigma_{2,22}^{-1},\Sigma_{3,12}\Sigma_{3,22}^{-1}]\!]
\\+&
(\mY_{121}-\mM_{121})\times_2(\Sigma_{2,12}\Sigma_{2,22}^{-1})
-
[\![(\mY_{212}-\mM_{212});\Sigma_{1,12}\Sigma_{1,22}^{-1},I_{n_2},\Sigma_{3,12}\Sigma_{3,22}^{-1}]\!]
\\+&
(\mY_{112}-\mM_{112})\times_3(\Sigma_{3,12}\Sigma_{3,22}^{-1})
-
[\![(\mY_{221}-\mM_{221});\Sigma_{1,12}\Sigma_{1,22}^{-1},\Sigma_{2,12}\Sigma_{2,22}^{-1},I_{n_3}]\!],
\end{split}
\end{equation*}
and 
$\varphi_q(u)
=\mbE\left(R^2_{q}\varphi_{\rtU}(u)\mid\rtY_{\backslash\{111\}}=\mY_{\backslash\{111\}}
\right),
$ 
with 
$
R^2_{q} = 
D^2_{\Sigma_{3,\{11\}|\bullet}\otimes \Sigma_{2,\{11\}|\bullet}\otimes \Sigma_{1,\{11\}|\bullet}}(\rtY_{111},\allowbreak\mM_{111}),
$.
If $\varphi(u) = \exp(-u/2)$, then the conditional distribution of
$R^2_{q}$ given $\rtY_{\backslash\{111\}}=\mY_{\backslash\{111\}}$ is
  $\chi^2_{n_1n_2n_3}$, and $\varphi_q(u) = \exp(-u/2)$.

\subsection{The Singular EC Wishart distribution}\label{sec:singularecwishart}

Suppose $\rmX\sim \EC_{m,n}(0,\Sigma,I_n,\varphi)$ for $n<m$. Then $\rmS=\rmX\rmX^\top\sim\ECW_m(n,\Sigma,\varphi)$ but $\rmS$ is singular.
In the following we extend the results of \citet{srivastavaandkhatri79,srivastava03} to obtain the joint PDF of the functionally independent elements of $\rmS$.
\begin{theorem}\label{thm:swishart}
Suppose $\rmS\sim \ECW_m(n,\Sigma,\varphi),$ with $ n<m$,  and for the $n\!\times\! n$ matrix $\rmS_{11}$,
$$
\rmS = \begin{bmatrix}
\rmS_{11} & \rmS_{12}\\ \rmS_{12}^\top & \rmS_{22}
\end{bmatrix}.
$$ 
Then the joint PDF of the functionally independent elements of $\rmS$ is 
\begin{equation}
f_{[\rmS_{11}\rmS_{12}]}(\sfS)=
\dfrac{\pi^{n^2/2}}{\Gamma_n(n/2)|\Sigma|^{n/2}}|\sfS_{11}|^{(n-m-1)/2}g(\tr(\Sigma^{-1}\sfS)).
\end{equation}
\end{theorem}
\begin{proof}
See Section \ref{proof:swishart}.
\end{proof}


\section{Supplement to Section 4}

\subsection{Linear and quadratic tensor-variate classification}\label{sec:classmethod}

Consider $G$ tensor-valued populations $\pi_1,\pi_2,\dots,\pi_G$. A
Bayes-optimal rule that minimizes the total probability of
missclassification assigns an observed tensor $\mX$ to group
$\pi_{i^*}$ if $i^* = \argmax_{i=1,\dots,G}\eta_if_i(\mX)$, where
$f_i(\mX)$ is the PDF of $\pi_i$ evaluated at $\mX$, and $\eta_i$ is
the prior probability that a member of $\pi_i$ gets correctly
classified \citep{andersonandbahadur62,thompsonetal20}.  

We first briefly visit the two-class problem. Here, we consider two tensor-valued populations $\pi_1$ and $\pi_2$. For $i,j=1,2$ let $\eta_i$ be the prior probability that a member of $\pi_i$ gets correctly classified to $\pi_i$, and also let $\mathbb{P}(i|j)$ be the probability that a member of $\pi_j$ gets misclassified to $\pi_i$ $(i\neq j)$. The total probability of misclassification (TPM) in this case is $\mathbb{P}(2|1)\eta_1 + \mathbb{P}(1|2)\eta_2$.
Now, suppose we are interested in classifying one tensor-valued
observation $\mX$ to $\pi_1$ or $\pi_2$. In this case, a Bayes-optimal
rule that minimizes the TPM assigns $\mX$ to $\pi_1$ if
$\eta_1f_1(\mX)\geq \eta_2f_2(\mX)$, where $f_i(\mX)$ is the PDF of
group $\pi_i$ evaluated at $\mX$
\citep{andersonandbahadur62,thompsonetal20}.

The extension of this rule to the case where there are $G$ groups $\pi_1,\pi_2,\dots,\pi_G$ assigns $\mX$ to  $\pi_{i^*}$ if $i^* = \argmax_{i=1,2,\dots,G}\eta_if_i(\mX)$, where $f_i$ and $\eta_i$ correspond to the PDF and prior probability that a member of $\pi_i$ gets correctly classified, respectively.

Suppose that  $\pi_i$ has the
$\EC_{\bbm}(\mM_i,\Sigma_{i,1},\sdots,\Sigma_{i,3},\varphi)$ distribution with PDF
\eqref{ellipt:PDF1}. Then an observation $\mY$ is assigned to $\pi_1$
over $\pi_2$ if 
\begin{equation}\label{eq:eqda_general}
\eta_1|\Sigma_1|^{-1/2} g( D^2_{\Sigma_1}(\mY,\mM_1))
\geq
\eta_2|\Sigma_2|^{-1/2} g ( D^2_{\Sigma_2}(\mY,\mM_2)),
\end{equation}
where $\Sigma_i = \otimes_{k=p}^1\Sigma_{i,k}$. This rule is greatly simplified for the TVN case where $g(u)=\pi^{-m/2}\exp(-u/2)$ as
\begin{equation}\label{eq:lqda_TVN}
\log\left(\dfrac{\eta_1^2}{\eta_2^2}\right) 
\!-\!\log \left(\dfrac{|\Sigma_1|}{|\Sigma_2|}\right)
\!-\!\left(D^2_{\Sigma_1}(\mY,\mM_1) - D^2_{\Sigma_2}(\mY,\mM_2)\right)
\!>\!0.
\end{equation}
\eqref{eq:lqda_TVN} is a quadratic discrimination analysis (QDA)
classification rule, with the usual reduction to linear discriminant
analysis (LDA) for  homogeneous covariance structures in the case of GSM
distributions with PDF as in \eqref{dengengamma}.

The parameters involved in \eqref{eq:eqda_general} can be estimated
from a training dataset using the methodology of Section
\ref{sec:MLE}. Moreover, our ToTR methodology of Section
\ref{sec:totr_ec} allows us to estimate the mean parameters $\mM_1$
and $\mM_2$ under low-rank formats, thereby facilitating substantial parameter reduction in many imaging applications.  ML estimates under the homogeneity of variance can be obtained by fitting the  TANOVA model
\begin{equation}\label{eq:totr_lda}
\mY_i = \langle \bx_i | \mB\rangle + \mE_i
,\quad 
\mE_i\sim \EC_{\bbm}(0,\sigma^2\Sigma_1,\dots,\Sigma_p,\varphi),
\end{equation}
where $\bx_i=[0,1]'$ or $\bx_i=[1,0]'$ accordingly as whether
$\mY_i\in \pi_1$ or $\mY_i\in\pi_2$. With population-specific
covariances, we  estimate the parameters in \eqref{eq:eqda_general} by
fitting separate EC
models to the training data from each group using the methods of
Section \ref{sec:MLE}. For the ToTR under low-rank format case,
fitting models to separate groups can be thought as the special case
of  \eqref{eq:totr_lda} when $\bx_i$ is the scalar 1. We now
illustrate our  methodology on images from the AFHQ database.

In section \ref{sec:classification} we display the sequentially applied Radon and DWT transforms
of the cat and dog images displayed in Figure
\ref{subfig:catndogimag}. These transformed figures are displayed in
Figure \ref{fig:radondwt_animals} and are examples of the $\mY$s that
are used in the calculation of the decision rule and in the
prediction.

\begin{figure}[!htb]
\includegraphics[width=\linewidth, page= 1]{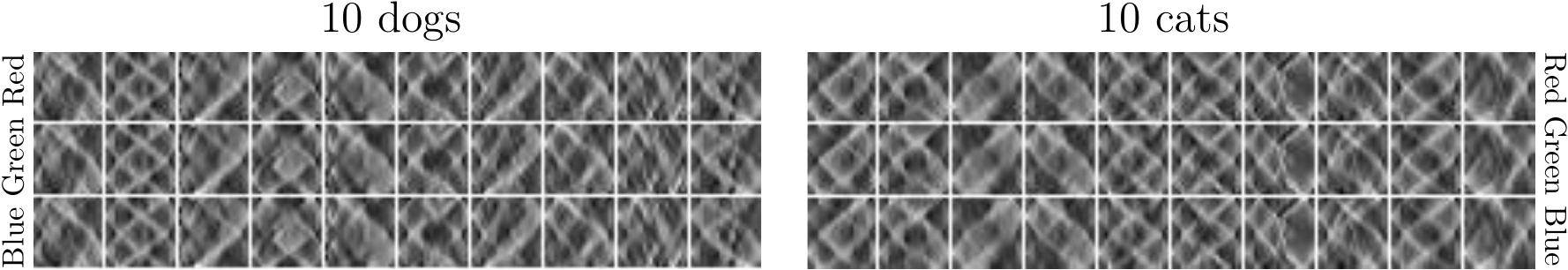}
  \centering
  \caption{Radon and DWT transforms across the 3 RGB channels for the sample dog and cat images of Figure \ref{subfig:catndogimag}.}
  \label{fig:radondwt_animals}
\end{figure}

\section{Proofs}
\subsection{Proof of Theorem \ref{thm:reshap_EC}}\label{proof:reshap_EC}

\begin{proof}
From the invariance of the Mahalanobis distance under tensor
reshapings, $\rtY,\rtY_{(k)},\rtY_{<k>}$ and $\vecc(\rtY)$ all have
the same characteristic function (CF), equivalently written for each case, as
\begin{equation*}
\begin{split}
\psi_{\rtY}(\mZ) 
\overset{1}{=}&
\exp(\I\langle \mZ , \mM\rangle)\varphi(\langle \mZ , [\![\mZ ; \Sigma_1,\Sigma_2,\hdots,\Sigma_p ]\!]\rangle)
\\\overset{2}{=}&
\exp(\I\langle \mZ_{(k)} , \mM_{(k)}\rangle)\varphi(\langle \mZ_{(k)} , [\![\mZ_{(k)} ; \Sigma_k,\Sigma_{-k}  ]\!]\rangle)
\\\overset{3}{=}&
\exp(\I\langle \mZ_{<k>} , \mM_{<k>}\rangle)
\varphi(\langle \mZ_{<k>} , [\![\mZ_{<k>} ; \otimes_{q=1}^k \Sigma_q,\otimes_{q=k+1}^p\Sigma_{q} ]\!]\rangle)
\\\overset{4}{=}&
\exp(\I\langle \vecc(\mZ) , \vecc(\mM)\rangle)\varphi(\langle \vecc(\mZ) , [\![\vecc(\mZ) ; \Sigma ]\!]\rangle). 
\end{split}
\end{equation*}
\end{proof}

\subsection{Proof of Theorem \ref{ellipt:distform}}\label{proof:ellipt:distform}
\begin{proof}  

The CF of $[\![\rtY; \sfA_1,\sdots,\sfA_p ]\!]$ can be
  expressed in terms of the CF of $\rtY$ as  $\psi_{\rtY}([\![\mZ;
  \sfA_1^\top,\sdots,\sfA_p^\top]\!])$, and using  
  \eqref{ellipt:cf1}, we have
$$
\psi_{\rtY}([\![\mZ; \sfA_1^\top,\sdots,\sfA_p^\top]\!]) \!=\!
\exp(\I\langle \mZ , [\![\mM; \sfA_1,\sdots,\sfA_p ]\!]\rangle)
\varphi(\langle \mZ , [\![\mZ ; \sfA_1\Sigma_1\sfA_1^\top, \sdots,\sfA_p\Sigma_p\sfA_p^\top ]\!]\rangle).
$$
The rest of Part \ref{reshap:par1} of the theorem follows from \eqref{ellipt:cf1}.
\end{proof}

\subsection{Proof of Theorem \ref{ellipt:conditional}}\label{proof:ellipt:conditional}
\begin{proof}
 From Part~\ref{reshap:par4} of Theorem \ref{thm:reshap_EC}, we can write the distribution of $\vecc{(\rtY)}$ as
\begin{equation*}
\begin{split}
\vecc{(\rtY)}  = 
\begin{bmatrix}	
\vecc(\rtY_1)\\
\vecc(\rtY_2)
\end{bmatrix}
 \sim
\EC_{m_{-p}\times (n_1+n_2)}\Big(
\begin{bmatrix}
\vecc(\mM_1)\\
\vecc(\mM_2)
\end{bmatrix},
\begin{bmatrix}
\Sigma_{11} \otimes \Sigma_{-p} & \Sigma_{12} \otimes \Sigma_{-p}\\
\Sigma_{21} \otimes \Sigma_{-p}  & \Sigma_{22} \otimes \Sigma_{-p}
\end{bmatrix},
\varphi
 \Big).
 \end{split}
\end{equation*}
The remainder of the proof follows from Corollary 5 in
\citep{cambanisetal81} and Part~\ref{reshap:par1} of Theorem \ref{thm:reshap_EC}, after noticing that $
\vecc(\mM_{\{1\}|\{2\}}) = \vecc(\mM_1)+ (\Sigma_{12}\otimes\Sigma_{-p})(\Sigma_{22}\otimes\Sigma_{-p})^{-1}(\vecc(\mY_2) -\vecc(\mM_2))
$ and $
\Sigma_{p,\{11\}|\bullet} \otimes \Sigma_{-p} 
= (\Sigma_{11}\otimes\Sigma_{-p})
-(\Sigma_{12}\otimes\Sigma_{-p})(\Sigma_{22}\otimes\Sigma_{-p})^{-1}(\Sigma_{21}\otimes\Sigma_{-p}).
$
\end{proof}

\subsection{Proof of Theorem \lowercase{\ref{expectelement}}}\label{app:proofmoments}
\label{app:proofexpectelement}
The CF of $\rtY$ in \eqref{ellipt:cf1} can be written as 
  $
\psi_{\rtY}(\mZ) 
=
\exp(\I h(\mZ))\varphi(g(\mZ)),
 $
 where 
$
h(\mZ) = \langle \mZ , \mM\rangle
$ and $
g(\mZ) = \langle \mZ , [\![\mZ ; \Sigma_1,\hdots,\Sigma_p ]\!]\rangle.
$
Denote $\mZ_{\bbi} = \mZ(i_1,\hdots,i_p)$, the first derivative of $h(\cdot)$ as
$\partial h(\mZ)/\partial\mZ_{\bbi} = m_{\bbi}$, the first derivative of $g(\cdot)$ as $
 g_{\bbi} \overset{.}{=} (\partial h(\mZ))/(\partial\mZ_{\bbi})$,
 and the second derivative of $g(\cdot)$ as
 $ g_{\bbi\bj} \overset{.}{=} \partial^2 h(\mZ)/\partial\mZ_{\bbi}\partial\mZ_{\bj}  = 2\sigma_{\bbi\bj}.
 $
All higher order derivatives are zero. Then
\begin{enumerate}
\item The first moment is obtained from 
$
\mbE(Y_{\bbi})  = \I^{-1} \dfrac{\partial \psi_{\rtY}}{\partial\mZ_{\bbi}}(0)
$,
 where
\begin{equation*}
\dfrac{\partial \psi_{\rtY}(\mZ)}{\partial\mZ_{\bbi}}
=
\exp(\I h(\mZ))\big[
g_{\bbi}\varphi'(g(\mZ)) + \I m_{\bbi}\varphi(g(\mZ))
\big].
\end{equation*}

\item The second moment uses
$
\mbE(Y_{\bbi}Y_{\bj})  = \I^{-2} \frac{\partial^2  \psi_{\rtY}}{\partial\mZ_{\bbi}\partial\mZ_{\bj}}(0)
$, where $\frac{\partial^2 \psi_{\rtY}}{\partial\mZ_{\bbi}\partial\mZ_{\bj}}(\mZ)$ is

$
\exp(\I h(\mZ))\times
\Bigg[
\varphi''(g(\mZ))\big(
g_{\bbi}g_{\bj}
\big)
+
\varphi'(g(\mZ))\big(
g_{\bbi\bj} + 
\I g_{\bbi}m_{\bj} + \I g_{\bj}m_{\bbi}
\big) +
\varphi(g(\mZ))\big(
\I^2m_{\bbi}m_{\bj}
\big)
\Bigg].
$

\item The third moment  is obtained using $
\mbE(Y_{\bbi}Y_{\bj}Y_{\bk})  = \I^{-3} \frac{\partial^3  \psi_{\rtY}}{\partial\mZ_{\bbi}\partial\mZ_{\bj}\partial\mZ_{\bk}}(0)
$, where $\frac{\partial^3 \psi_{\rtY}}{\partial\mZ_{\bbi}\partial\mZ_{\bj}\partial\mZ_{\bk}}(\mZ)$ is

$
\exp(\I h(\mZ))
\Bigg[
\varphi'''(g(\mZ))
\big(
g_{\bbi}g_{\bj}g_{\bk}
\big)
+
\varphi(g(\mZ))\big(
\I^3m_{\bbi} m_{\bj}m_{\bk}
\big) \\+
\varphi''(g(\mZ))\big(
g_{\bj\bk}g_{\bbi} + g_{\bbi\bk}g_{\bj}
 + g_{\bbi\bj}g_{\bk} +\I g_{\bj}g_{\bk}m_{\bbi} 
 + \I g_{\bbi}g_{\bk}m_{\bj}
+ \I g_{\bbi}g_{\bj}m_{\bk}
\big)
+\varphi'(g(\mZ))\big(
\I g_{\bbi\bj}m_{\bk} + \I g_{\bbi\bk}m_{\bj} + \I g_{\bj\bk}m_{\bbi} + 
\I^2 m_{\bj}m_{\bk}g_{\bbi} + \I^2m_{\bbi}m_{\bk}g_{\bj} + \I^2m_{\bbi}m_{\bj}g_{\bk}
\big)
\Bigg].
$

\item We get the fourth moment from $
\mbE(Y_{\bbi}Y_{\bj}Y_{\bk}Y_{\bl})  = \I^{-4} \dfrac{\partial^4  \psi_{\rtY}}{\partial\mZ_{\bbi}\partial\mZ_{\bj}\partial\mZ_{\bk}\partial\mZ_{\bl}}(0)
$, where (ignoring terms involving the first derivative of $g(\cdot)$ because they are proportional to elements in $\mZ$, which will be set to 0) $\frac{\partial^4 \psi_{\rtY}}{\partial\mZ_{\bbi}\partial\mZ_{\bj}\partial\mZ_{\bk}\partial\mZ_{\bl}}(\mZ)$ is
\begin{equation*}
\begin{split}
&\exp(\I h(\mZ))\times
\Bigg[
\varphi''(g(\mZ))
\big(
g_{\bbi\bj}g_{\bk\bl}+
g_{\bbi\bk}g_{\bj\bl}+g_{\bbi\bl}g_{\bj\bk}
\big)+
\I^4\varphi(g(\mZ))\big(
m_{\bbi}m_{\bj}m_{\bk}m_{\bl}
\big)
\\&
+
\I^2\varphi'(g(\mZ))\big(
g_{\bbi\bj}m_{\bk}m_{\bl} + g_{\bbi\bk}m_{\bj}m_{\bl} +
 g_{\bbi\bl}m_{\bj}m_{\bk} + 
g_{\bj\bk}m_{\bbi}m_{\bl} +
 g_{\bj\bl}m_{\bbi}m_{\bk} + g_{\bk\bl}m_{\bbi}m_{\bj}
\big)
\Bigg].
\end{split}
\end{equation*}
\end{enumerate}

\subsection{Proof of Theorem \lowercase{\ref{thm:moments_ec}}}
\label{proof.b2}
\label{proof:moments_ec}
\begin{proof} 
In these proofs, LHS and RHS refer to the left and right hand
sides of an expression.
  \begin{enumerate} 
\item The proof follows directly from Part~\ref{expectelement:1} of
  Theorem \ref{expectelement}.
\item Let $\rtX \sim
  \EC_{m_1,m_2,\hdots,m_p}(0,\Sigma_1,\Sigma_2,\hdots,\Sigma_p,\varphi)$. Then,
  $\text{Var}(\vecc(\rtY)) =  \text{Var}(\vecc(\rtX))$, and from
  Part~\ref{expectelement:2} of 
  Theorem \ref{expectelement} and  \eqref{defvec}, we have
\begin{equation*}
\begin{split}
\text{Var}(\vecc(\rtX)) = & \sum_{\substack{i_1,\dots,i_p\\j_1\dots j_p}}
  \mbE[\mX(\bbi)\mX(\bj)]
  (\bigotimes_{q=p}^1\be_{i_q}^{m_q})
  (\bigotimes_{q=p}^1\be_{j_q}^{m_q})'
  \\=& 
  -2\varphi'(0)
  \bigotimes_{q=p}^1
   \sum_{i_q,j_q}
  \left(\Sigma_q(i_q,j_q)\be_{i_q}^{m_q}{\be_{j_q}^{m_q}}'\right)
  = 
  -2\varphi'(0) \bigotimes_{q=p}^1 \Sigma_q.
\end{split}
\end{equation*}  
The following moment follows from Parts~\ref{par51} and~\ref{par52} of
this theorem, and is used in Parts~\ref{par54} and~\ref{par55}.
\begin{equation}\label{eq:secmom}
\begin{split}\mbE\big[(\vecc \rtY)(\vecc \rtY)'\big]
=
(\vecc \mM)(\vecc \mM)'-2\varphi'(0)(\bigotimes_{k=p}^1 \Sigma_k).
\end{split}
\end{equation}

\item We can write the LHS as
$
\tr\big\{(\bigotimes_{k=p}^1\sfA_k)\mbE\big[(\vecc \rtY)(\vecc \rtY)'\big]\big\}.
$
The result follows from  \eqref{eq:secmom} after using $\tr(\bigotimes_k \sfA_k)=\prod_k\tr(\sfA_k)$.

\item The vectorization of the LHS can be expressed as
$
\mbE\big[(\vecc \rtY)(\vecc \rtY)'\big](\bigotimes_{k=p}^1\sfA_k)(\vecc\mV).
$
Using \eqref{eq:secmom} results in
$\langle \mV,[\![\mM;\sfA_1,\sfA_2,\hdots,\sfA_p]\!]\rangle\vecc(\mM)
 -2\varphi'(0)\vecc[\![\mV;\Sigma_1 \sfA_1',\hdots,\Sigma_p \sfA_p']\!]. 
$

 We have proved that the vectorizations of the LHS and the RHS are the same, and since they are tensors of the same size, they must be the same.
\item  Let $\rtX \sim \EC_{m_1,m_2,\hdots,m_p}(0,\Sigma_1,\Sigma_2,\hdots,\Sigma_p,\varphi)$. Then the first and third moments of $\rtX$ are zero, and using $\rtY \overset{d}{=} \rtX + \mM$, the LHS can be 
expressed as
\begin{equation*}
\begin{split}
\langle \mM,[\![
&
\mM;\sfA_1,\hdots,\sfA_p]\!]\rangle\mM +
 \mbE(\langle \rtX,[\![\rtX;\sfA_1,\hdots,\sfA_p]\!]\rangle) \mM  
 \\+&
 +
  \mbE(\langle \mM,[\![\rtX;\sfA_1,\hdots,\sfA_p]\!]\rangle\rtX) +
\mbE(\langle \mM,[\![\rtX;\sfA_1',\hdots,\sfA_p']\!]\rangle\rtX). 
\end{split}
\end{equation*}
The statement of this part of the theorem follows from
Parts~\ref{par54} and~\ref{par55}.
\item If $\rtX \sim \Sph_{m_1,m_2,\hdots,m_p}(\varphi)$ and $(\sfC_k,\sfD_k)=(\Sigma_k^{1/2}\sfA_k\Sigma_k^{1/2},\Sigma_k^{1/2}\sfB_k\Sigma_k^{1/2})$ for all $k=1,2,\hdots,p$, where $(.)^{1/2}$ is the symmetric square root, then the LHS can be written as
\begin{equation}\label{eq:alt1}
\sum_{\bbi,\bj,\bk,\bl}
\mbE(\bX_{\bbi}\bX_{\bj}\bX_{\bk}\bX_{\bl})
(c_{\bbi\bj}d_{\bk\bl} ),
\end{equation}
where $\bX_{\bbi} = \rtX(i_1,i_2,\hdots,i_p)$ and $c_{\bbi\bj}=\prod_{q=1}^p \sfC_q(i_q,j_q)$, $d_{\bk\bl} = \prod_{q=1}^p \sfD_q(k_q,l_q)$.
Using Part~\ref{expectelement:4} of Theorem \ref{expectelement}, we know that 
$
\mbE(\bX_{\bbi}\bX_{\bj}\bX_{\bk}\bX_{\bl}) 
=
4\varphi''(0)(
1_{\bk\bl}1_{\bj\bbi}\!+\!
1_{\bj\bl}1_{\bbi\bk}\!+\!
1_{\bbi\bl}1_{\bj\bk}
),
$
where $1_{\bk\bl} = 1$ only if $i_k=j_k$ for all $k=1,\hdots,p$, and
it is $0$ otherwise. From these, we can write \eqref{eq:alt1} based on
\begin{equation*}
\begin{split}
\sum_{\bbi\bj}
(
c_{\bbi\bbi}d_{\bj\bj} + 
c_{\bbi\bj}d_{\bbi\bj} +
c_{\bbi\bj}d_{\bj\bbi}
) =
\prod_{k=1}^p \big[ \tr(\sfC_k)\tr(\sfD_k) \big]+ 
\prod_{k=1}^p \tr(\sfC_k\sfD_k)
+
\prod_{k=1}^p \tr(\sfC_k\sfD_k').
\end{split}
\end{equation*}
The statement follows after replacing $(\sfC_k,\sfD_k)$
 with their original expressions.
\end{enumerate}
\end{proof}

\subsection{Proof of Theorem \ref{thm:tenswish1}}\label{proof:tenswish1}

\begin{proof} We prove the two parts of the theorem in turn:
\begin{enumerate}
\item 
For any $i=1,2,\hdots,n$, the distribution of each $\rtY_i$ is obtained from Theorem 
\ref{ellipt:distform} with $\rtY_i = \rtY\times_{p+1}{\be_i^n}^\top$. They are uncorrelated using Part \ref{expectelement:2} of Theorem \ref{expectelement} because whenever $i_{p+1} \neq j_{p+1}$,
$
\mathbb{E}\left[ 
\rtY(i_1,i_2,\hdots,i_{p+1})\rtY(j_1,j_2,\hdots,j_{p+1})
\right] = 0.
$
\item From Part \ref{reshap:par3} of Theorem~\ref{thm:reshap_EC} we obtain that $\rtY_{<p>}\sim \EC_{m,n}(0,\Sigma,I_n,\varphi)$. The remainder of the theorem follows from Definition \ref{def:tenswish} after using Lemma \ref{lemma:outervectens} to obtain that
$$
\left(\sum\limits_{i=1}^n (\rtY_i \circ \rtY_i)\right)_{<p>}
=
\sum\limits_{i=1}^n (\vecc \rtY_i)(\vecc \rtY_i)^\top
=
\rtY_{<p>}\rtY_{<p>}^\top.
$$
\end{enumerate}
\end{proof}


\subsection{Proof of Theorem \ref{mluncorrelated}}\label{proof:mluncorrelated}

\begin{proof}
We have from Part \ref{reshap:par3} of our Theorem~\ref{thm:reshap_EC} that $\rtY_{<p>}\sim \EC_{m,n}(\mM_{n<p>},\sigma^2\Sigma,I_n,\varphi)$, where the $m\times n$ matrix $\mM_{n<p>} = [\vecc(\mM)\dots \vecc(\mM)]$. Therefore, from Theorem 1 of \cite{andersonandetal86} we have that the ML
estimator of $\vecc(\mM)$ is $\vecc(\tilde{\mM})$, and that of 
$
\sigma^2I_n\otimes \left[\bigotimes_{k=p}^1\Sigma_k\right]
$ is
$
(nm/d_g)\tilde{\sigma}^2I_n\otimes \left[\bigotimes_{k=p}^1\tilde{\Sigma}_k\right].
$ 
This means that in our case, the ML estimator of the covariance matrix
is the same as in the TVN case but up to a constant of proportionality, and since the scale matrices $\Sigma_k$ are proportionally constrained as $\Sigma_k(1,1)=1$ for all $k=1,2,\dots ,p$, it follows that the ML estimator of $\Sigma_k$ is $\tilde\Sigma_k$, and the proportionality constant ${nm}/d_g$ is absorbed by  $\tilde\sigma^2$.
\end{proof}

\subsection{Proof of Theorem \ref{thm:tyrobustML}}\label{proof:tyrobustML}
\begin{proof}
First, let $Z_{ik} = S_{ik} /||\mY_i||^2$.
 Then the loglikelihood in \eqref{eq:tyfull} can be written in terms of $\Sigma_k$ only as
\begin{equation}\label{eq:ty_liklK}
\ell(\Sigma_k) = \dfrac{nm_{-k}}{2}\log|\Sigma_k^{-1}| - \dfrac{m}{2} \sum_{i=1}^n \log\tr(\Sigma_k^{-1}Z_{ik}),
\end{equation}
and has matrix derivative
\begin{equation}\label{eq:derL}
\dfrac{\partial\ell}{\partial\Sigma_k^{-1}} = \dfrac{nm_{-k}}{2}\Sigma_k - \dfrac{m}{2} \sum_{i=1}^n \dfrac{S_{ik}}{\tr(\Sigma_k^{-1}S_{ik})},
\end{equation}
where the equality follows because
$
S_{ik}/\tr(\Sigma_k^{-1}S_{ik}) = Z_{ik}/\tr(\Sigma_k^{-1}Z_{ik}).
$
Setting \eqref{eq:derL} to 0 leads to the following fixed-point iterative procedure for obtaining $\Sigma_k$:
$
\widehat\Sigma_k^{(t+1)} = \dfrac{m_k}{n}
\sum_{i=1}^n \dfrac{S_{i,k}}{\tr(\widehat\Sigma_k^{-1(t)}S_{ik})}. 
$
However, we are interested in the constrained optimization under $\Sigma_k(1,1)=1$. To do this, we write the equality constraint function $g$ and its matrix derivative as
\begin{equation}\label{eq:ty_gK}
g(\Sigma_k) = {\be_1^{m_k}}'\Sigma_k\be_1^{m_k} - 1
,\quad
\dfrac{\partial g}{\partial\Sigma_k^{-1}} =
- \boldsymbol{\sigma}\boldsymbol{\sigma}'
\end{equation}
where $ \boldsymbol{\sigma}$ is the first column of $\Sigma_k$ and $\be_1^{m_k}\in \mathbb{R}^{m_k}$ has 1 at position 1 and zero everywhere else.
 Our equality constraint function is zero only if $\Sigma_k(1,1)=1$.
With $\ell(\cdot)$ and $g(\cdot)$ as in  \eqref{eq:ty_liklK} and \eqref{eq:ty_gK} we write the Lagrange multiplier as
$
\mL(\Sigma_k,\lambda)=
 \ell(\Sigma_k) - \dfrac{\lambda}{2} g(\Sigma_k),
$
and its derivative with respect to $\Sigma_k$ follows directly from  \eqref{eq:derL} and \eqref{eq:ty_gK}. Based on this derivative, the constrained optimization must satisfy
\begin{equation}\label{eq:tyroot}
\dfrac{nm_{-k}}{2}\Sigma_k - \dfrac{m}{2} \sum_{i=1}^n \dfrac{S_{ik}}{\tr(\Sigma_k^{-1}S_{ik})} + \lambda
\boldsymbol{\sigma}\boldsymbol{\sigma}' = 0.
\end{equation}
In  \eqref{eq:derL} and \eqref{eq:ty_gK}, we ignored the symmetric structure of $\Sigma_k$ when taking derivatives because it is simpler and  leads to the same root in \eqref{eq:tyroot}.
This equation is of the form of Equation (B.1) in \citet{glanzandcarvalho18}, and therefore we know that it is satisfied for \eqref{eq:tytensfixed}.
\end{proof}

\subsection{Proof of Theorem \ref{lemma:outervectens}}\label{proof:outervectens}

\begin{proof} 
In this proof we denote $\rtX(i_1,\hdots,i_p) = \rtX(\bbi)$,
$\rtY(j_1,\hdots,j_q) = \rtY(\bj)$, and we use the definitions of
vectorization, canonical matricization and outer product in Equations
\eqref{defvec}and \eqref{canonical}
\begin{equation*}
\begin{split}
(\rtX\circ\rtY)_{<p>} 
&=\sum\limits_{\substack{i_1,\sdots i_p\\j_1\sdots,j_p}}
\rtX(\bbi)\rtY(\bj)
\left(
\left( \circs\limits_{r=1}^p \boldsymbol{e}_{i_r}^{m_r} \right) \circ
\left( \circs\limits_{r=1}^q \boldsymbol{e}_{j_r}^{n_r} \right)\right)_{<p>}
\\&
= \sum\limits_{\substack{i_1,\sdots i_p\\j_1\sdots,j_p}}
\rtX(\bbi)\rtY(\bj)
\left(
\left(\bigotimes\limits_{r=p}^1 \boldsymbol{e}_{i_r}^{m_r}\right) 
\left(\bigotimes\limits_{r=q}^1 \boldsymbol{e}_{j_r}^{n_r}\right)' \right)
\\&=
\sum\limits_{i_1,\sdots,i_p}
\left(\rtX(\bbi)
\bigotimes\limits_{r=p}^1 \boldsymbol{e}_{i_r}^{m_r}\right) 
\sum\limits_{j_1,\sdots,j_p}
\left(\rtY(\bj)
\bigotimes\limits_{r=q}^1\boldsymbol{e}_{j_r}^{n_r}\right)'
\\&=
(\vecc \rtX)(\vecc \rtY)'.
\end{split}
\end{equation*}
\end{proof}

\subsection{Proof of Theorem \ref{thm:swishart}}\label{proof:swishart}
\begin{proof}
In this case, we have that $\rmS_{22} =
\rmS_{12}^\top\rmS_{11}^{-1}\rmS_{12}$, and therefore the functionally
independent elements of $\rmS$ are in $\rmS_{11}$ and
$\rmS_{12}$. From Definition \ref{def:matwish}, we have that there
exists an $\rmX\sim \EC_{m,n}(0,\Sigma,I_n,\varphi)$ such that $\rmS
\overset{d}{=} \rmX\rmX^\top$, and from \eqref{ellipt:PDF1}, we have that the pdf of $\rmX$ is
\begin{equation}
f_{\rmX}(\sfX)=|\Sigma|^{-n/2}g(\tr(\Sigma^{-1}\sfX\sfX^\top)).
\end{equation}
Transforming the random variable $\rmX$ to $[\rmS_{11}\rmS_{12}]$ can be performed using the singular value decomposition in \citep{srivastava03} or the triangular factorization method in \citep{srivastavaandkhatri79}. The remaining steps mirror those in \citep{srivastavaandkhatri79,srivastava03} in their development of the singular Wishart distribution, and so are omitted here. However, they ultimately involve multiplying $f_{\rmX}(\sfX)$ by
\begin{equation}
\dfrac{\pi^{n^2/2}}{\Gamma_n(n/2)}|S_{11}|^{(n-m-1)/2}.
\end{equation}
\end{proof}

\bibliographystyle{IEEEtran}
\bibliography{tensorontensor}

\end{document}